\documentclass{llncs}  
\usepackage{amsmath,amssymb}
\usepackage{graphicx}
 \usepackage{float}
\usepackage{fullpage}
 \floatstyle{boxed}
 \newfloat{algo}{h}{loa}
\floatname{algo}{{\textnormal{Algo}}}

\usepackage[ruled,linesnumbered]{algorithm2e}

\def\lc{\ensuremath{\texttt{lc}}}

\def\prem{\ensuremath{\mathsf{prem}}}
\def\coeff{\ensuremath{\mathsf{coeff}}}

\def\l{\ensuremath{\langle}}
\def\r{\ensuremath{\rangle}}

\def\Q{\ensuremath{\mathbb{Q}}}

\def\R{\ensuremath{\mathbb{R}}}

\def\calC{\ensuremath{\mathcal{C}}}
\def\calD{\ensuremath{\mathcal{D}}}
\def\calI{\ensuremath{\mathcal{I}}}
\def\Tree{\ensuremath{\mathcal{T}}}

\def\xgcd{\ensuremath{\mathrm{xgcd}}}

\def\LargestFactor{\ensuremath{\mathrm{largestFactor}}}

\def\nilpotentFactor{\ensuremath{\mathrm{nilpotentFactor}}}
\def\WeierstrassMonic{\ensuremath{\mathrm{WeierstrassMonic}}}
\def\subres{\ensuremath{\mathrm{modifiedSubres}}}

\def\fm{\ensuremath{\mathfrak{m}}}
\def\fa{\ensuremath{\mathfrak{a}}}

\newcommand{\jump}[1]{\ensuremath{[\![#1]\!]} }

 \def\bfS{\ensuremath{\mathbf{S}}}

 \newcommand{\x}[1]{\ensuremath{x_1,\ldots,x_{#1}}}

\newtheorem{Lem}{Lemma}
\newtheorem{Rmk}{Remark}
\newtheorem{ex}{Example}

\begin{document}

\begin{frontmatter}
  \title{Computation of gcd chain over the power of an irreducible polynomial} 
%%\titlerunning{Bit-size reduction of triangular sets in two or three variables}
%% \subtitle{[Extended abstract]}

 \author{Xavier Dahan\thanks{extended version of a presentation made at the conference ADG'2018 (September~12-14, Nanning, China)}} %% \and Marc Moreno Maza\inst{2}}

   \institute{Ochanomizu university,
     Faculty of General Educational Research \& Dept. of Mathematics\\
     \email{xdahan@gmail.com}
 %%     \and
 %%     University Of Western Ontario, Dept. of Computer Science, ORCCA\\
 %%     \email{moreno@uwo.ca}
   }

%% \author{Xavier Dahan}
%%   \address{Ochanomizu university (Tokyo),
%% Faculty of General Educational Research\& Dept. of Mathematics}
%%   \ead{dahan.xavier@ocha.ac.jp,\ xdahan@gmail.com}
%%

\maketitle

\begin{abstract}
A notion of gcd chain has been introduced by the author at ISSAC 2017 for two univariate monic polynomials
with coefficients in a ring $R = k[\x{n}]/\l T\r$ where $T$ is a {\em primary} triangular set of dimension zero.
A complete algorithm to compute such a gcd chain remains  challenging. This work
treats completely the case of a triangular set $T=(T_1(x))$ in one variable, namely a power of an irreducible polynomial.
This seemingly ``easy'' case reveals the main steps necessary for treating
the general case, and it allows to isolate the particular one step that
does not directly extend and requires more care.

%% In~\cite{Da17} the first author introduced a gcd notion for two univariate polynomials with coefficients modulo a primary triangular ideal
%% $\fq$ of dimension zero, based on the unique factorization modulo $\fq$.
%% It was  coined {\em gcd-chain}.  The primary ideal is supposed to be 
%% given by a lexicographic Gr\"obner basis (lexGb) in a pure triangular form, hence a triangular set $\l T\r$.
%% While the notion and existence are proved, algorithms given are limited or rely on a strong assumption.
%% This article studies in more details this algorithm,
%% especially by showing how to iterate subresultant
%% algorithm when the first nilpotent remainder is met with Weierstrass preparation theorem.
%% This also identifies when an exact generalization of the ideal equality 
%% $\l a,b \r =\l g\r$ through the proposed isomorphism fails to exist;
%% this corrects an inaccurate   statement of the Definition of gcd-chain.
%% For example, for a triangular set in one variable such an isomorphism always exists.
%% A preliminary implementation is also presented that treats when the first nilpotent
%% subresultant can be put into the form of application of the Weierstrass preparation theorem.
%% This as well as another issue, that of the recovery of the precision loss that follows, are  let for 
%% future work.
\end{abstract}
%% 

%% 
%% \begin{keyword}
%% Triangular set, gcd
%% \end{keyword}
%%  

\end{frontmatter}

\section{Introduction}
Computing gcd is without a doubt   one of the most fundamental algorithm  in computer algebra
and computational aspects have been  studied extensively, till today.
%%Still today,
%every year several articles treat of this topic, would it be about approximate gcds, the related subresultant algorithms etc.
In~\cite{Da17} is introduced the concept of {\em gcd chain} to bring a similar notion of the classical gcd 
of polynomials of one variable  over a field, to the case of over a ring $R:=k[\x{n}]/\l T \r$ where $T$ is a primary triangular set
of dimension zero.
Such a ring has nilpotent  elements, and non nilpotent  elements are invertible.
Some attempts to treat this case in prior~\cite{Da17} have concluded in somewhat  unsatisfactory solutions.
Indeed,
%% for polynomials of one variable over an ideal that is ``almost'' maximal, here a primary ideal,
a desirable fundamental property of gcd is an ideal equality $\l a,b\r = \l g\r $. While if $a$ and $b$ have coefficients
in such a ring of type $R$, % =k[\x{n}]/\l T\r$ where $T$ is a {\em primary} triangular set,
it is well-known that a polynomial $g$ does not exist in general, the outcome of~\cite{Da17}
being to present a strategy to circumvent this impediment by ``iterating'' somehow a Pseudo Remainder
Sequence when a nilpotent remainder is met.
On the algorithmic side, this raises several challenging questions, even in the seemingly ``easy'' case 
of a primary triangular set $T=(T_1(x_1))$ of {\em one} variable. As this article shows,
this case is already not simple. And it is  important since it builds the framework to tackle the case of 
several variables. In particular we identify that all steps, except one that requires more work,
extend to more than one variable.

\subsection{Motivation} An early motivation in computing gcds
over triangular sets come from the {\em triangular-decomposition} algorithm
to solve polynomial (commutative or differential) systems~\cite{wang2012elimination,chenMMM2012}.
This set of computational methods traces back to the early work of Ritt~\cite{Ri32},
and the major computational advances realized later by Wu-Wen Tsu~\cite{WWT86}.
This has lead  to several new directions of researches, followed by many researchers.
In term of algorithms, only {\em pseudo-divisions} were initially used.
In 1993, in~\cite{Ka93} Kalkbrenner introduced a ``gcd''-point of view to realize
the decomposition, and the elimination (See also the notes~\cite{Hu1}).
This point of view has later been significantly developed by M. Moreno-Maza {\em et al.}
%% essentially by the second author and collaborators,
in particular with the implementation of the   library {\tt RegularChains}~\cite{LeMMMXi04} in the software  Maple.

However, such a gcd does not handle ``faithfully'' polynomials having multiplicities;
this question was raised as early as 1995~\cite{MMMOo95} and later studied furthermore
in ~\cite{liMMM09}, but without a satisfactory general answer.
In this regard, the present work situates in the realm of triangular decomposition
as initiated by Wu-Wen Tsu.

The gcd chain has the following {\em geometric} interpretation.
The underlying  triangular set can be thought as some algebraic constraints,
over which one may want to compute with further polynomials,
that is over the solutions of the constraints {\em only}.
It may happen that the solution (a constraint),
is {\em multiple}. One may think of $t_1(x)=x^3$ for example,
in the case where  constraints are modeled by a polynomial
of one variable like in this work.
When computing {\em over} $t_1$, this allows to consider Taylor expansions at order 2
(for example to control the first and second order derivative of the solution$=$constraint being modeled).

In Example~\ref{ex:x^3} below, we want to compute the constraints defined by both polynomials $a$ and $b$
with coefficients in $\R[x]/\l t_1\r = \R[x]/\l x^3\r$. We obtain three cases, 
displayed in Figure~\ref{fig:three},
as computed by the algorithm of this paper. 
%% However these may depend if one consider
%% only position (precision $x$), position and speed (precision $x^2$) or position, speed acceleration  (precision $x^3$).
%% 
\begin{figure}\label{fig:three}
\includegraphics[height=5cm]{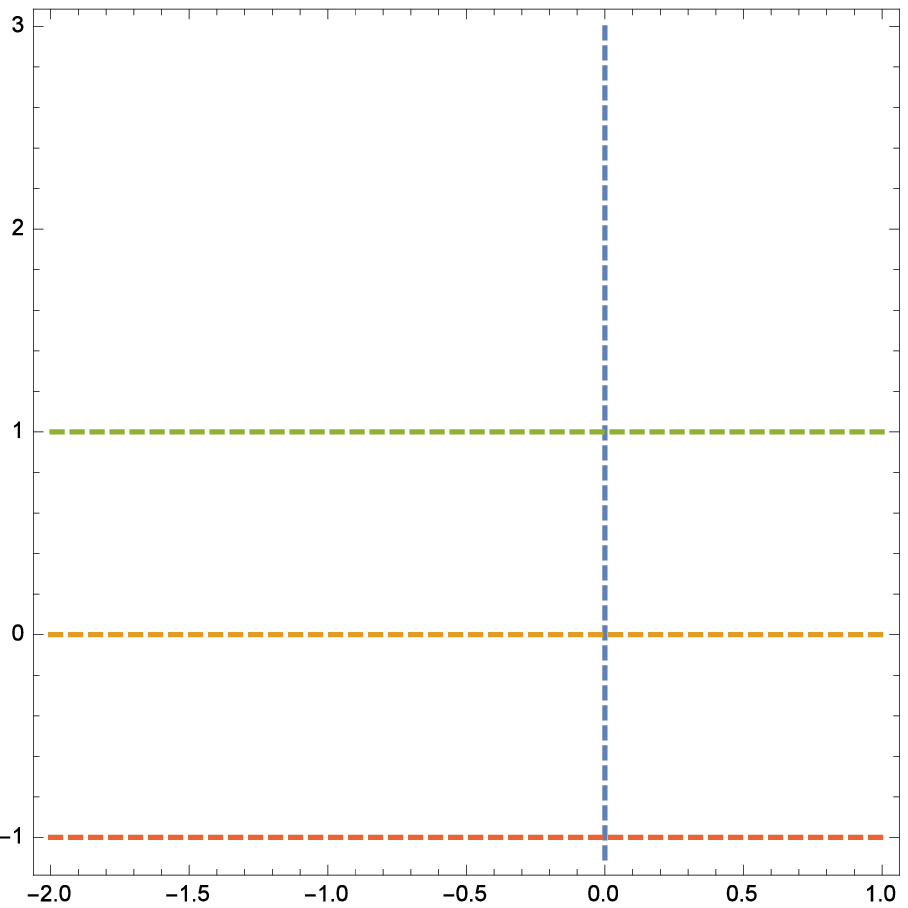}
\ \ \includegraphics[height=5cm]{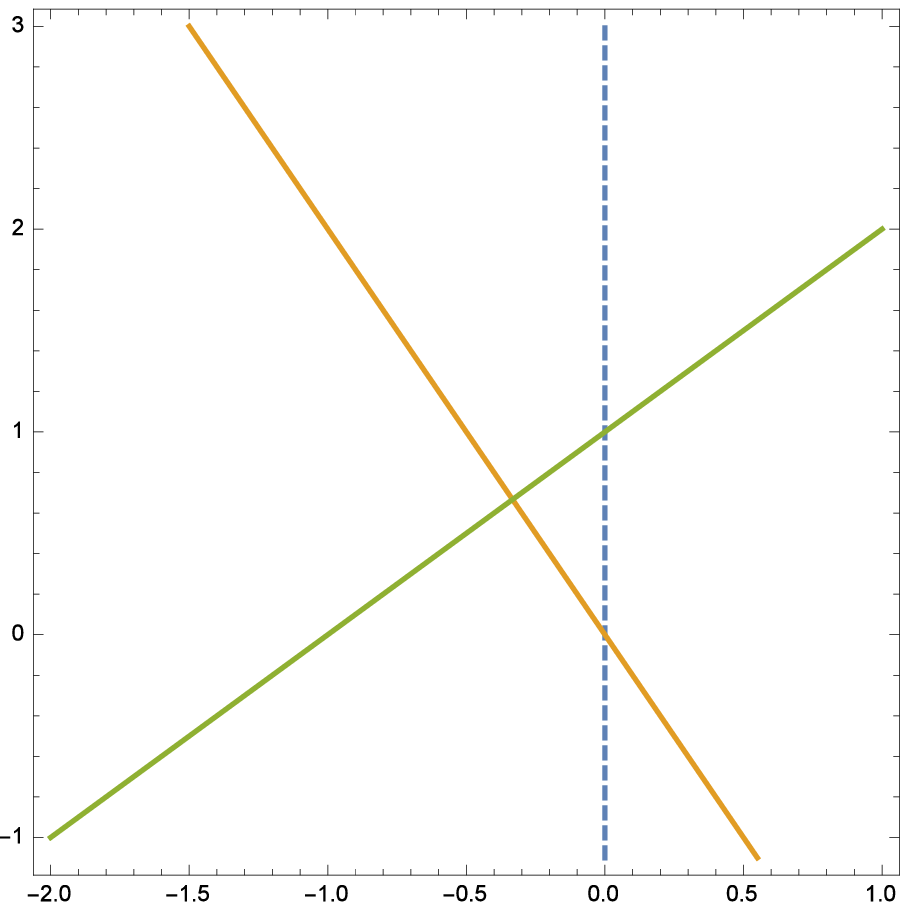}
\ \ \includegraphics[height=5cm]{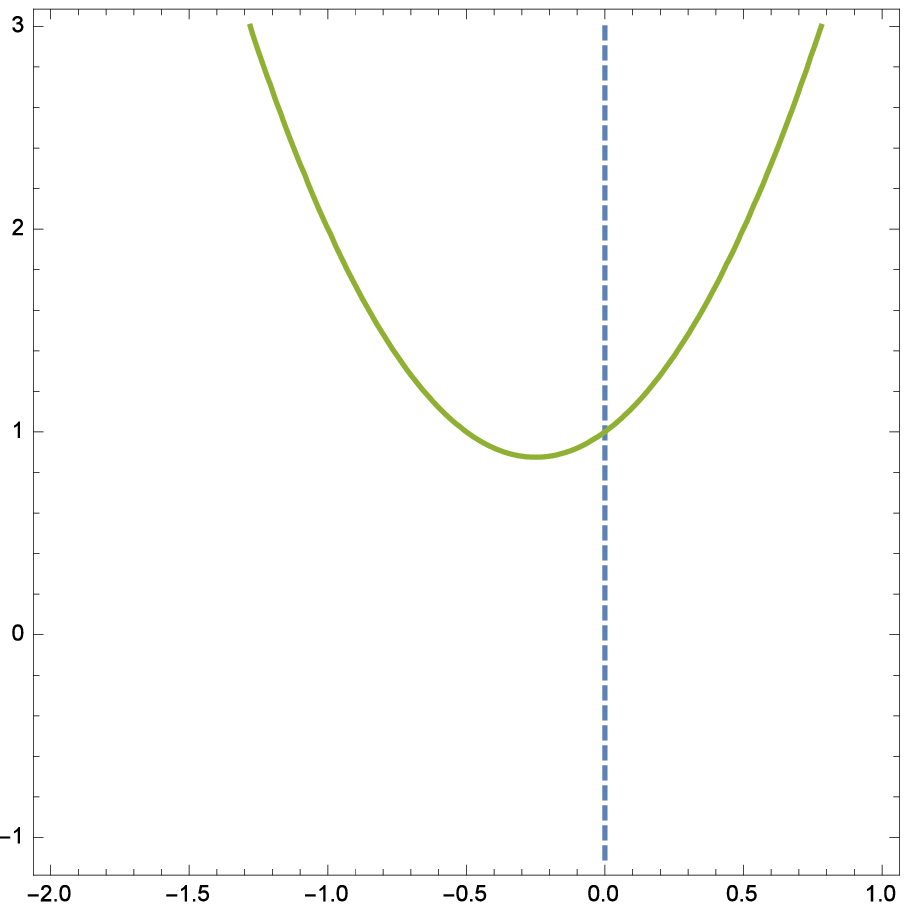}
\caption{Precision $x$ (left): Three points intersection of $y=0,-1,1$ with $x=0$. Precision $x^2$ (middle): Two lines, expanded from $y=0,1$. Precision $x^3$ (right):
One parabola, expanded from $y=1$}
\end{figure}

\subsection{Definitions}
A primary triangular set in one variable is just a power of an irreducible polynomial $p$:
$T=(T_1(x)) = (p(x)^e)$. The ring $R=k[x]/\l T\r $ is local of
maximal ideal $\fm=\l p\r$. It is therefore {\em Henselian}: monic polynomials admit a unique factorization into coprime factors
(not into irreducibles: see~\cite[\S~3.1]{Da17} for more details). 
Before giving the technical definitions, a connection with the classical gcd over
unique factorization domains will enlighten the differences with the new notion.

Since $R$ is Henselian  we can write the unique factorization of the input polynomials $a$ and $b$ into coprime factors as follows:
\begin{equation}\label{eq:factorab}
a =  a_1\cdots   a_\gamma \cdot a_{\gamma+1} \cdots a_\alpha,  \qquad b =  b_1\cdots b_\gamma \cdot b_{\gamma+1} \cdots b_\beta.
\end{equation}

We have ordered the factors so that $\gcd(a_i \bmod \l p\r , b_j \bmod \l p\r)=\rho_i^{\nu_i} $ where $\rho_i \in (k[x]/\l p\r)[y]$ is an irreducible
polynomial; And  $\nu_i  = \min (\lambda_i  , \delta_i)$ where $a_i =\rho_i^{\delta_i}+\cdots$ and $b_i=\rho_i^{\lambda_i} +\cdots$ (so
that $a_i \equiv \rho_i^{\delta_i} \bmod \l p\r $ and $b_i \equiv \rho_i^{\lambda_i} \bmod \l p\r $).
If $a$ and $b$ have some factors $a_\ell$ and $b_m$ for $\ell,m >\gamma$ then 
we assume that $\gcd(a_\ell \bmod \l p\r , b_m \bmod \l p\r)=1$.
\begin{ex}\label{ex:factors}
We reproduce the example of~\cite[Ex. 3.5, Ex. 5.2]{Da17}.
Consider the polynomials $a$ and $b$ along with their unique factorizations into coprimes modulo $T=x^3$.
{\footnotesize \begin{align*}
a&=a_1 a_2 a_3 = ((y+1)^2 + x (2 y + 1) + x^2 (y+1)) (y+2 x + 3 x^2) (y- 1 - x -2 x^2)\\
b& =b_1 b_2 b_3 = (y+1 +  2 x + x^2) (y+2 x +4 x^2) (y-1-x -2 x^2)
\end{align*}}
We have $\gamma=3$ and in each case $\gcd(a_i \bmod \l x\r ,b_i \bmod \l x\r)$ is non trivial.
\end{ex}

The informal discussion that follows is rigorously detailed in  Sections~4.1-4.2 of~\cite{Da17}.
Under this point of view, it is convenient to refer to the more common terminology of unique factorization domains,
but there is a caveat:  ``precision''.
Modulo $p$, the gcd of $a$ and $b$ (over the field $k[x]/\l p\r$) is well-defined and is as expected:
$g=\gcd ( a_i \bmod \l p\r , b_j \bmod \l p\r) = \prod_{i=1}^\gamma \rho_i^{\nu_i}$.
But there is no isomorphism $ \l a,b,T\r = \l g,T\r$, only the more coarse $\l a,b,p\r = \l g,p\r$ holds.
A more refined notion of ``common factors at certain precision'' allows to obtain an ideal equality $\l a,b,T\r$.
\begin{definition}\label{def:precision}
A monic polynomial $\ c\ $ is said to be a {\em common factor} of $a$ and $b$ {\em at precision} $\ell$ iff the Euclidean divisions
of $\ a\ $ and $\ b\ $ by $\ c\ $ have both zero remainder modulo $p^\ell$,
but at least one  non-zero modulo $p^{\ell+1}$.
$$
a = c a' + r_a ,\quad  b =  c b' + r_b \quad  \Rightarrow \quad  r_a \ , \ r_b \equiv 0 \bmod \l p^\ell \r, 
\quad  r_a {\rm~or~}  r_b \not\equiv 0 \bmod \l p^{\ell+1}\r.
$$
\end{definition}
Let $\calI_1$ be the set of indices of common factors of $a$ and $b$ at the smallest precision $e_1$. 
Let $\calI_2$ be the set of indices of common factors $a$ and $b$ at the next to smallest precision $e_2$, etc.
we  obtain a partition  $(\calI_1,\calI_2,\ldots,\calI_s)$ of the set of common factors of $a$ and $b$ 
and the associated precision exponents $[e_1,\ldots,e_s]$.

\begin{ex}(Example~\ref{ex:factors} continued)
We notice that $e_1 = 1$ and $\calI_1  = \{1\}$. Indeed $y+1$ is the largest common divisor of $a_1$ and $b_1$,
and it is modulo $x^{e_1}=x$ ; There is  no  common divisor of $a_1$ and $b_1$
modulo $x^{e_1+1}=x^2$.
Next we observe that $e_2 = 2$  and $\calI_2=\{2\}$: the largest common divisor of $a_2$ and $b_2$ is
$y+2 x $ and it is modulo $x^{e_2}=x^2$ (there is no common divisor of $a_2$ and $b_2$
modulo $x^{e_2+1}=x^3$). Last modulo
$x^e= x^3$ the factors $a_3$ and $b_3$ have a common divisor implying that $\calI_3=\{e_3\}$.
\label{ex:sequence}\end{ex}

Sections~4.1-4.2 of~\cite{Da17}
prove the {\em existence and uniqueness}, in  the case of a triangular set of one
variable,  of the tuple of
indices $(\calI_1,\calI_2,\ldots,\calI_s)$
and of the sequence of increasing precision powers $[e_1,e_2,\ldots,e_s]$.
Before that,  with the notations of
Eq.~\eqref{eq:factorab}, denote $G_j^{(a)} := \prod_{j \in \calI_j} a_i$ and
$G_j^{(b)} := \prod_{j\in \calI_j} b_i$, so that:
\begin{equation}\label{eq:a=Gia}
a = \left( \prod_{i=1}^s G_i^{(a)}\right) \cdot a_{\gamma+1}\cdots a_{\alpha},
\qquad b = \left( \prod_{i=1}^s G_i^{(b)}\right) \cdot b_{\gamma+1}\cdots b_{\beta}.
\end{equation}
And from the discussion above 
both $G_i^{(a)} $ and  $ G_i^{(b)}$ 
have one maximal common factor which is at precision $p^{e_i}$.
Let us write it $G_i$ (since ehere it comes alone,
it makes no harm to think of $G_i$ as ``the'' gcd  and to write $G_i \equiv \gcd( G_i^{(a)} ,\ G_i^{(b)}) \bmod p^{e_i}$.
There exist monic polynomials $c_i^{(a)}$ and $c_i^{(b)}$ both relatively
prime modulo $p$, such that: 
\begin{equation}\label{eq:Gia}
G_i^{(a)} \equiv  c_i^{(a)} G_i \bmod \l p^{e_i}\r ,\quad  \text{and}
\quad  G_i^{(b)} \equiv  c_i^{(b)} G_i \bmod \l p^{e_i}\r,
\end{equation}
We obtain an  equality of ideals:
$\l a,b, T\r = \l G_1,p^{e_1}\r \cdots \l G_s ,  p^{e_s}\r$.
The formal definition of gcd chain modulo a triangular set of one variable is as follows:

\begin{definition}\label{def:gcdChain}
Given two monic polynomials $a,b\in R[y]$,  a {\em gcd chain} of $a,b$
is  a sequence $(g_i , p^{e_i})_{i=1,\ldots,s}$ such that:
\begin{itemize}
\item  $e_1 <\ldots <e_s \le e$ and $\deg_y (g_1) >\cdots >\deg_y(g_s)$.
\item $g_i$ is the product of all common factors of $a$ and $b$ at precision $\ge e_i$. 
\item $g_{i+1}$ divides $g_{i}$ modulo $p^{e_{i+1}}$. Defining $G_{i}:=g_{i}/g_{i+1} \bmod p^{e_{i+1}}$  for  $i=1,\ldots , s-1$  and $G_s:=g_s$,
the following isomorphism holds:
\begin{equation}\label{eq:isochain}
R[y]/\l a,b\r \simeq (k[x]/\l p^{e_1} \r) [y]/\l G_1\r \times \cdots 
 \times 
(k[x]/\l p^{e_s} \r) [y]/\l G_s\r 
\end{equation}
where the r.h.s is a direct product of rings. We have moreover for all $i=1,\ldots,s$:
\begin{equation}\label{eq:block}
(k[x]/\l p^{e_i}\r)[y]/\l a,b\r \simeq (k[x]/\l p^{e_1} \r) [y]/\l G_1\r 
\times \cdots \times (k[x]/\l p^{e_{i-1}} \r) [y]/\l G_{i-1}\r 
\times (k[x]/\l p^{e_{i}} \r) [y]/\l g_{i}\r
\end{equation}
\end{itemize}
\end{definition}

\begin{ex}(Examples~\ref{ex:factors},~\ref{ex:sequence} continued) \label{ex:x^3}
Let $p(x)=x$, $T=(T_1(x)) =(x^3)=(p^3)$.
%% To fix notations $\fq=\l T\r= \l x_1^3\r$,
%% $\fp=\l p_1\r =\l x_1\r=  \sqrt{\l T\r}$.
The two monic polynomials $a$ and $b$ when expanded are written: 
{\small
$$\begin{array}{rcl}
a & =& y^4 + \left(2 x^2+3 x +1\right) y^3+\left(-x^2 - x -1\right) y^2
%% & + & 
\left(-13 x^2-4 x -1\right) y-7 x^2-2 x\\
b & =&   y^3 + (3 x^2 + 3 x) y^2 + (-3 x^2 - 3 x - 1) y - 10 x^2 - 2 x
\end{array}
$$
}
According to the discussion made in Example~\ref{ex:sequence},
  $e_1=1,\ e_2=2,\ e_3=3$. The  gcd chain is given by:
$$
[ ((y-1) y (y+1) ,\  x) \ ,\ ((y-1-x)(y+2 x ) ,\ x^2)\ , \ (y-1-x - 2 x^2,\ x^3)],
$$
and yields the  following isomorphism according to~\eqref{eq:isochain}.
\begin{multline}
(k[x]/\l x^3\r)[y]/\l a ,b \r \simeq 
(k[x]/\l x\r)[y]/\l y+1\r \times 
(k[x]/\l x^2\r)[y]/\l y+2 x \r\\  \times
(k[x]/\l x^3\r)[y]/\l y-1-x - 2 x^2 \r
\end{multline}
\end{ex}
However, Section~5 of~\cite{Da17} dealing with algorithms  
is more an indication of directions for future work, than a complete and definitive exposition.
This is what the present work does, treating
the case of one variable completely.

\subsection{Related works\label{ssec:related}}
First of all, let us clarify what may appear as 
an elephant in  the room:
\begin{center}
why not computing the squarefree part of $T=p^e$ ?
\end{center}
The first reason is that sticking with $T=p^e$ really allows to discover the gcd-chain. Otherwise
starting from the gcd known at precision $p$
% but which of its factors are known at higher precision is then more
% complicated to find. 
one could perform Hensel lifting and, roughly speaking,  see at each step which parts
still divide $a$ and $b$ etc.
But then quadratic convergent lifting may  not be convenient enough to 
recover each precision $e_i$, in that it may miss some.
Some refinements should then be devised\ldots with additional costs.
The linearly convergent Hensel lifting is not as efficient.
The second  reason is  for considerations of generalization to  a primary triangular set
of {\em several} variables.
Computing the radical is still possible (see {\em e.g.} \cite{mou2013decomposing,LiMouWang2010})
at a reasonable cost, but this time recovering the factors at the required ``precision''
becomes far more complicated, assuming it is possible.
Consequently, for convenience and  generalization in mind,
\begin{center}
the squarefree part operation is not considered.
\end{center}
Gcds over non-radical triangular sets have been addressed in~\cite{MMMOo95,liMMM09}.
But none provided a structural isomorphism that mimics the case of gcd of polynomials over
a field. In a different direction, some works have focused on representing not
only solution points, but their {\em multiplicity} as well~\cite{Cheng2014,Li2003multiplicities,Marcus2012,Alvandi2015}. The multiplicity is a much
coarser information than
what provides an ideal isomorphism; Moreover the methods proposed in these works are
not simpler than the one proposed in~\cite{Da17} and here. See Section ``Previous Work'' of~\cite{Da17}
for more details.
Compared to the algorithmic Section~5 of~\cite{Da17}
the present article gives a complete treatment, with  Hensel lifting; in particular
the over-optimistic over-simplistic Assumption~(C) made therein is elucidated
with Weierstrass preparation theorem.

When the input irreducible polynomial is $p(x) = x$,  the article~\cite{moroz2016fast}
which computes a ``truncated'' resultant of bivariate polynomials $a$ and $b$
shares some common features with the present work.
However, the gcd-chain carries more information than the classical resultant,
and it is not clear wether the quite intricate algorithm that is devised in~\cite{moroz2016fast}
can be adapted.
Therein, a generalized version of  the  ``half-gcd'' is proposed whose purpose
is a better asymptotic complexity.
 Another significant
difference is the use of Weiertrass factorization to cope with nilpotent leading coefficients.
This notion is more natural and efficient than the ``normalization'' lemma~1 of~\cite{moroz2016fast};
It does not require Hensel lifting and to know the squarefree part $p$ of the input $T=p^e$.
%% 
%% See the Conluding Remarks section for more about the resultant
%% computation with the present algorithm.
%% 

Besides, the ``clarity'' of the classical subresultant based  study presented here 
allows to detect that  all tasks used in the algorithm extend
quite straightforwardly
to primary triangular sets of more than one variable,
except one:
 Algo.~\ref{algo:nilFac}  ``$\nilpotentFactor$''  which ``removes'' nilpotent part.
See Section ``Concluding remarks'' for more about this.

\subsection{Organization of the paper}
Section~\ref{sec:subres} introduces the core routine ``largestFactor''
that executes one iteration of the subresultant algorithm.
It finds the ``last non-nilpotent subresultant'',
uses the Weierstrass factorization to make it, together with the first nilpotent subresultant,
monic before another iteration.
The complete algorithm ``gcdChain'' that calls ``largestFactor''
is presented and analyzed in Section~\ref{sec:gcdChain}.
The last section~\ref{sec:Weier}
details the variation of the Weierstrass preparation theorem that we have used. 
Some concluding remarks end the article.

\subsection*{Acknowledgement} The idea of  Weierstrass preparation theorem
was transmitted to me by M. Moreno-Maza. I am grateful
for this precious advice and other discussions we had in March 2018.

\section{A subresultant based main algorithm\label{sec:subres}}

The outcome of this section is Algorithm~\ref{algo:largestFactor}
``largestFactor'' presented in Subsection~\ref{ssec:largestFactor}.
It introduces  Weierstrass factorizations
at Lines~\ref{enum:WeierA} and~\ref{enum:WeierB}.
%% , as well as Hensel lifting at Line~\ref{enum:HL}
In the first subsection   several subroutines
used in this algorithm are introduced first.

\subsection{Preliminary routines}
%%The first subroutine 

%% catches the last non-nilpotent subresultant
%% in a subresultant sequence.
%% 
%% \begin{algo}
%% \SetAlgoLined
%% \SetAlgoNoEnd
%% \DontPrintSemicolon
%% \KwIn{Subresultant sequence $\bfS=[S_{r_0},S_{r_1},\ldots ,S_{r_t}]$
%%  and a power of a irreducible polynomial  $T=(p(x)^e)$.}
%% \KwOut{index $j$ such that $S_{r_j}$ is not nilpotent, and $S_{r_{j+1}}$ is nilpotent}
%% 
%% 
%% $j=t$\;
%% 
%%  \While(\tcp*[f]{is leading coefficient $\lc(S_{r_j})$ nilpotent $\bmod p^e$ ?})
%% {$\Res_{x}( {\rm lc} (S_{r_j}),T)=0 $ in $ k$ \nllabel{line:isNil}}
%% {$j=j-1$}
%% \Return{$j+1$}
%% \caption{\label{algo:lastNil} index of the last non nilpotent subresultant}
%% \end{algo}
%% 
%% \begin{proof}[Correctness of Algo~\ref{algo:lastNil}]
%% It uses the Block structure of a subresultant sequence.
%% In our context, this has been sort out in Lemma~2 of~\cite{liMMM09}.
%% (see also 2nd paragraph of ``Correctness of Algo. 1'' in~\cite[p.~115]{Da17}).
%% \hfill $\Box$ \end{proof}

Algorithm~\ref{algo:nilFac} computes the largest power
of $p$ that divides all the coefficients of a polynomial
in $(k[x]/\l p^e\r)[y]$.  It is a key routine in the present work.

\begin{algo}
\SetAlgoLined
\SetAlgoNoEnd
\DontPrintSemicolon
\KwIn{$T = p^e$, power an irreducible polynomial. Nilpotent polynomial $F = F_0 + F_1 y +\cdots + F_r y^r $ modulo $T$}
\KwOut{$P=p^\ell$, where $\ell \le e$, and  $p^\ell\mid F$ but $p^{\ell+1}\nmid F$\\
Largest degree $i$ such that $F_i$ is invertible if exists; $-1$ otherwise. 
}

$P \leftarrow T$\ , \ \ $i \leftarrow r$ \;

\While{$i> -1$ and $\deg(P)>0$}{
  $P\leftarrow \gcd(F_i,P)$\; ;  \quad  $i\leftarrow i-1$\;
}
\Return{$P, i$}
\caption{\label{algo:nilFac} nilpotentFactor: find the largest power of $p$ that divides $F$,  and $T$}
\end{algo}

\begin{proof}[Correctness of Algo.~\ref{algo:nilFac}]
  For any value of $i$, one has $P = \gcd(F_r , F_{r-1}, \ldots,F_i,T)$. Since $T = p^e$, we can write $p^{\ell_i} := \gcd(F_r,\ldots,F_i,T)$.
  The sequence $\{ \ell_i\}_i$ decreases. Let $\ell$ be the minimal value. By definition $p^\ell\mid F$ since $p^\ell \mid F_j$ for all $j$,
  and $\ell$ is the largest integer having this property. Note that the output $P$ is precisely $p^\ell$.
On the other hand, if $\deg(P) = 0$, then $\gcd(F_i, p^{\ell_{i-1}}) = 1$
and $\ell_{i-1}>0$ since otherwise the while loop would have been exited before. We see
that $F_i$ is the coefficient of $F$ of largest degree that is invertible --- when it exists.
\hfill $\Box$ \end{proof}

Let $P,i$ be the ouput of Algorithm~\ref{algo:nilFac}. The first statement in the remark hereunder follows directly from the algorithm. The second one from
the classical fact that a polynomial is nilpotent iff all its coefficients are nilpotent. 
\begin{Rmk} If $\deg(P)>0$ then $i = -1$ and reciprocally. The polynomial $F$ is nilpotent iff $i=-1$.
\end{Rmk}

A  key subroutine Algo.~\ref{algo:WeierMonic} ``$\WeierstrassMonic$'',
is used to ``make monic'' a non-nilpotent polynomial. 
It is based on Corollary~\ref{cor:div}, which
details are postponed to Section~\ref{sec:Weier}.

\begin{algo}
\SetAlgoLined
\SetAlgoNoEnd
\DontPrintSemicolon
\KwIn{Primary triangular set $T=(p(x)^e)$\\
  Univariate polynomial $f=\cdots+f_k y^k +\cdots $ in $R[y]$ where $R=k[x]/\l p^e\r$,
  index $k$ of the  non nilpotent coefficient of $f$ of largest degree}

\KwOut{Monic polynomial $h=y^k + h_{k-1} y^{k-1} +\cdots+ h_0$ in the Weierstrass factorization theorem
(Corollary~\ref{cor:div})\\
In particular $\l f, T\r = \l h , T\r$ (ideal equality)}

$ (u p^e + v f_k=1) \leftarrow $ Extended Euclidean Algorithm of $p^e$ and $f_k$ in $k[x]$\;

$f_k^{-1} \leftarrow v$\; \tcp*[f]{inverse of $f_k$ modulo $p^e$}

$y^k = f g + r  \leftarrow$ Division of $y^k$ by $f$ following Corollary~\ref{cor:div} (using $v$), up to precision $O(y^{k+1})$\;

\Return{$y^k -r$}
\caption{\label{algo:WeierMonic} WeierstrassMonic}
\end{algo}

Last, the routine Algo.~\ref{algo:subres} below computes the subresultant
pseudo remainder sequence modulo $T$ and proceeds to the necessary 
modifications. All specifications are described within the algorithm.
%%Before, the following lemma studies the subresultant algorithm modulo $T$.

\begin{algo}
\SetAlgoLined
\SetKwProg{try}{try}{}{}
\SetKwProg{catch}{catch}{}{}
\SetAlgoNoEnd
\DontPrintSemicolon
\KwIn{Monic polynomials $a,b$ with $\deg_y(a)\ge \deg_y(b)>0$.\\
 Power of an irreducible polynomial  $T=(p(x)^\epsilon)$.}
\KwOut{ Subresultant pseudo-remainder sequence modulo $T$:   $\bfS=[S_{r_0}= a , S_{r_1}= b,  S_{r_2},  ,\ldots, S_{r_t},\ldots]$\\
Index $j$  such that $S_{r_j}$ is not nilpotent and $S_{r_{j+1}}$ is nilpotent or zero\\
Integer $i$ such that $\coeff(S_{r_j},i)$ is the non-nilpotent coefficient of largest degree\\
If $S_{r_{j+1}}$ is not zero, largest degree polynomial $P(x)=p(x)^\ell\in k[x]$ such that $P | S_{r_{j+1}}$.}
\BlankLine

\try{\nllabel{enum:try1}}{
  
  Compute a subresultant sequence modulo $T$: $ \bfS =[S_{r_0}= a , S_{r_1}= b,  S_{r_2},  ,\ldots, S_{r_t},\ldots]$ \nllabel{enum:sresT}\;

  $j+1\leftarrow $ index of last non-zero subresultant~\nllabel{enum:lastNonZero}\;
  
  $i_{new} \leftarrow -1$~;  $P_{new}\leftarrow 0$\;

}
\catch(\tcp*[f]{division failed at the $r_{t+2}$-th subresultant  \nllabel{enum:error}})
      {(``Error: Leading coefficient of $S_{r_{t+2}}$ is not invertible'')}
      {
        $j \leftarrow t+1$~; $P_{new}, i_{new} \leftarrow \nilpotentFactor (S_{r_j} , T)$~\nllabel{enum:18}\;
        
        \If(\tcp*[f]{$S_{r_j}$ is not nilpotent \nllabel{enum:nomoreNil}}){$i_{new} > -1$}{
          $A \leftarrow  \WeierstrassMonic(S_{r_{j}}  , T, i_{new})$ \nllabel{enum:AWeier}\;
          
          $\_\_ , i \leftarrow \nilpotentFactor (S_{r_{j-1}}, T)$ \hfill \tcp*[f]{$\coeff(S_{r_{j-1}},i)$ is invertible} \nllabel{enum:indexi}\;
          
          $B \leftarrow \WeierstrassMonic(S_{r_{j-1}} , T , i)$\nllabel{enum:BWeier}\;
          
          \Return{$\subres (A, B, T)$}
        }
      }
      
      \While(\tcp*[f]{$S_{r_j}$ is nilpotent \nllabel{enum:stillNil}}){$i_{new} = -1$}{
        $i_{old} \leftarrow i_{new}$~; $P_{old} \leftarrow P_{new}$~; $j \leftarrow j -1$\;
        
        $P_{new}, i_{new} \leftarrow \nilpotentFactor ( S_{r_j} , T)$\;
      }

\Return{$\bfS,\ j,\ i_{new},\ P_{old}$}  
\caption{\label{algo:subres} $\subres$: Modified subresultant pseudo-remainder sequence}
\end{algo}

\begin{Lem}[Correctness of Algo.~\ref{algo:subres}]
The output $\bfS, j, i, P$ satisfies the specifications mentionned in ``ouput''.
\end{Lem}
\begin{proof}

Step~1: The computation of the subresultant p.r.s. modulo $T$ raises no exception at Line\ref{enum:sresT}.
This means that the  full chain has been
successfully computed, in particular the last subresultant is zero. 
By definition of the index $j$ at Line~\ref{enum:lastNonZero},  its index is $r_{j+2}$   and $S_{r_{j+1}}$ is the last  non  zero (mod $T$).

It goes next to Line~\ref{enum:stillNil}. The ``return'' occurs exactly when  $i_{new}> -1$ or equivalently 
that $S_{r_j}$ is not nilpotent; Additionally that $S_{r_{j+1}}$ is nilpotent or zero.
The specification of Algorithm ``$\nilpotentFactor$'' tell  that
 $i_{new}$ is the largest degree of $S_{r_j}$ whose coefficient $c_{i_{new}}$ is not nilpotent.
Since a polynomial is nilpotent iff all its coefficients are nilpotent, this implies that
$S_{r_j}$ is not nilpotent. Since  the while loop (Line~\ref{enum:stillNil}) is bottom-up
on the subresultant chain, $j$ is the {\em first} index for which the corresponding subresultant is not nilpotent.
Therefore $S_{r_{j+1}}$ is nilpotent, and the specifications of Algorithm ``$\nilpotentFactor$'' 
imply that $i_{old}=-1$ and that $P_{old}\mid S_{r_{j+1}}$. This shows that 
all the requirements sateted in the ``output'' of Algo.~\ref{algo:subres} are satisfied.
\smallskip

Step~2: The computation of the subresultant p.r.s. failed, an exception is caught. Lines~\ref{enum:sresT}-\ref{enum:lastNonZero} are skept
to go directly to Line~\ref{enum:18}. The leading coefficient of $S_{r_{t}}$ is not invertible therefore the 
subresultant $S_{r_{t+2}}$ {\em cannot} be computed (mod $T$) following the (classical subresultant p.r.s.) formula:
$$
S_{r_{t+2}}(a,b) =   \pm  \frac{\prem(S_{r_t} , S_{r_{t+1}} )}{c_{t}^{r_t - r_{t+1}} \lc(S_{r_{t}})},
\quad 
\begin{cases} c_1 = 1, \quad r_i  = \deg(S_{r_i})\\
c_i =  \lc(S_{r_i})^{r_{i-1}-r_i} c_{i-1}^{r_i-r_{i-1}+1}.\\ 
\end{cases}
$$
Nonetheless, $S_{r_t}$ and $S_{r_{t+1}}$ have been both computed successfully.
Index $j$ is set to $t+1$. If $i_{new} > -1$ then $S_{r_{t+1}}$ is not nilpotent,
therefore neither is $S_{r_t}$. Thus the call to Algo. ``$\WeierstrassMonic$'' is valid
at Lines\ref{enum:AWeier} and ~\ref{enum:BWeier}. And the recursive call to ``$\subres$'' makes sense.
Eventually, no exception is raised, or $S_{r_j}$ is nilpotent (at Line~\ref{enum:18}).
Moreover the recursive call is consistent since $\l a, b,T\r = \l S_{r_{t}} , S_{r_{t+1}} , T\r = \l A ,B,T\r$.

If $i_{new}=-1$ then $S_{r_j}$ is nilpotent, but it may not be the first nilpotent subresultant.
It goes next to the while loop (line~\ref{enum:stillNil}) and the proof follows that of Step~1.
 $\Box$ \end{proof}

\subsection{Algorithm ``largestFactor''\label{ssec:largestFactor}}
Now that the subroutines are defined, we describe
the main algorithm of this section.
It builds upon
the subresultant algorithm,
where instead of computing the ``last non-zero'' one,
it computes the ``last non-nilpotent'' subresultant.
Besides the comments and specifications,
the proof of correctness discusses in details of the several
steps. 
\begin{algo}
\SetAlgoLined
\SetKwProg{try}{try}{}{}
\SetKwProg{catch}{catch}{}{}
\SetAlgoNoEnd
\DontPrintSemicolon
\KwIn{Monic polynomials $f,A$ with $\deg_y(f)\ge \deg_y(A)$. Power of an irreducible polynomial  $T=(p(x)^\epsilon)$.}
\smallskip

\KwOut{ $g,\ p^{\epsilon_1},\ B,\ p^\ell \text{ or } \_\_ $  where $g$ is  monic and $\l g ,  p^{\epsilon_1} \r = \l f ,A , p^{\epsilon_1}\r$~;\\
 $B=$``$end$''~; or $B \in k[y]$  monic, $\deg_y(g) \ge\deg_y(B) $~;\\
$p^\ell = p^{\epsilon - \epsilon_1}$ is used for iterative calls in Algorithm~\ref{algo:gcdChain}\;
}
%% and $p^\ell = p^{e-e_1}$ is the new triangular set for recursive call marked with the preision loss}
\medskip

 \If(\tcp*[f]{Finished: No iteration necessary}){$A=0$}{\Return{$f,\ T$,\ ``$end$'',\ $\_\_$}\nllabel{enum:enda}}

 \If(\tcp*[f]{Finished: No iteration necessary}){$\deg(A)=0$}{\Return{$A,\ T$,\ ``$end$'',\ $\_\_$}\nllabel{enum:endb}}

 $\bfS, j, i, p^{e_1} \leftarrow \subres (A, f,T)$ \hfill \tcp*[f]{$S_{r_j}$ is not nilpotent, $S_{r_{j+1}}$ is. $p^{e_1} \mid S_{r_{j+1}}$ }\nllabel{enum:28}\;

%% $p^{e_1} \leftarrow \nilpotentFactor(S_{r_{j+1}},T)$ \tcp*[f]{Extract nilpotent part}\nllabel{enum:pell}\;
 
  $g \leftarrow \WeierstrassMonic(S_{r_j},T,i)$ \hfill \tcp*[f]{Put $S_{r_j}$ in  monic form} \nllabel{enum:WeierA} \;

\If(\tcp*[f]{case where $e_1=e$}){$p^{e_1} = T$\nllabel{enum:ifzero}}{

  \Return{ $g,\ T,$ ``$end$''$,\ \_\_ $}\nllabel{enum:returnzero}
}

  $ S \leftarrow S_{r_{j+1}} / p^{e_1} $ \nllabel{enum:divbyp} \hfill \tcp*[f]{Precision loss of $\deg(p^{e_1})$; $S$ is no more nilpotent}
%% . Note the precision loss by $p^\ell$. Without Hensel lifting, the remaining computations would be done at precisoin $p^{e -e_1}$} \;

$p^{\ell} \leftarrow T/p^{e_1}=p^{e-e_1}$\nllabel{enum:33}\;

   $\_\_ , i \leftarrow \nilpotentFactor(S,p^\ell)$ \hfill \tcp*[f]{$\coeff(S,i)$ is invertible}\nllabel{enum:34}\;

  $B \leftarrow \WeierstrassMonic(S , p^\ell,i )$ \hfill \tcp*[f]{Put $S$ in monic form} \nllabel{enum:WeierB}\;

%% $B\leftarrow {\rm HenselLift} (S,A,b,T,p^{e_1})$ \hfill \tcp*[f]{Recover precision loss}  \nllabel{enum:Hensel}\;

\Return{ $g ,\ p^{e_1} ,\ B ,\ p^\ell$}\nllabel{enum:return0}
\caption{\label{algo:largestFactor} The Largest Common Factor ($\LargestFactor$)}
\end{algo}

\begin{Lem}[Correctness of Algo.~\ref{algo:largestFactor}]
\label{lem:correct}
The output of Algorithm~\ref{algo:largestFactor} verifies: $\l f,A,p^{e_1}\r=\l g,p^{e_1}\r$.
\end{Lem}

\begin{proof}
If the algorithm exits at Line~\ref{enum:enda} then $\l f,A,T\r = \l f ,0,T\r =\l f,T\r$.
%% The gcd chain has length 1.
It then returns $f, T,$``$end$'', hence  $g = f$ and $p^{e_1} = T$, and $\l f,T\r = \l g,p^{e_1}\r$.
The gcd chain has one block.

If the algorithm exits at Line~\ref{enum:endb} then $A$ is monic constant, hence equal to 1; therefore
 $\l f,A,T\r = \l A ,T\r = \l T,1\r=\l 1\r $.
It returns $A, T, $``$end$'', hence  $g = A$ and $p^{e_1} = T$, thereby $\l g,p^{e_1}\r=\l A,T\r$. 
Here too the gcd chain has one block.

Assume now that the algorithm ends at Line~\ref{enum:return0} or at Line~\ref{enum:returnzero}.
At Line~\ref{enum:28}, according to the specifications of Algo.~\ref{algo:subres},
 $S_{r_j}$ is not nilpotent and $S_{r_{j+1}}$ is either nilpotent or zero (mod $T$).
Moreover $p^{e_1}$ is the largest power of $p(x)$ that divides $S_{r_{j+1}}$, and $\coeff(S_{r_j} ,i)$
is the coefficient of largest degree that is invertible.
%% 
%% 
%%  
%% Lines~\ref{enum:try2}-\ref{enum:catch} computes the subresultant pseudo-remainder sequence
%% modulo $T$ of $g$ and $f$, following the classical Collins-Brown's formula. As it uses divisions
%% of some leading coefficients, it may raise an error since  modulo $T$  elements are not all invertible.
%% This error is caught and the computation of the subresultant is ended. 
%% It means that  this coefficient is nilpotent which is precisely what we are looking for.
%% That is why no additional treatment is necessary in the ``catch'' Line~\ref{enum:catch}.
%% 
%% 
%% At Line~\ref{enum:lastNil}, Specification of Algorithm~\ref{algo:lastNil} insures
%% that $S_{r_j}$ is the last non nilpotent subresultant, whence $S_{r_{j+1}}$ is nilpotent.
%% 
%% At Line~\ref{enum:pell} is removed the nilpotent part of $S_{r_{j+1}}$ equal to $p^{e_1}$.
We can apply Algorithm~\ref{algo:WeierMonic} ``WeierstrassMonic'' at Line~\ref{enum:WeierA}; It outputs
a monic polynomial $g$
such that $S_{r_j}= u\cdot g$ where $u$ is a unit in $(k[x]/\l p^{e_1} \r)[y]$.
In particular $\l g, p^{e_1}\r =\l S_{r_j} , p^{e_1} \r $ $(\ddagger)$.

If the test of Line~\ref{enum:ifzero} is satisfied, then $S_{r_j}$ is not only
the last non-nilpotent subresultant but also the last non-zero one:  classical
subresultant theory insures that $\l S_{r_j} \r = \l f, A\r$ modulo $T$ (that is $e_1=e$).
From the specifications of the ``WeierstrassMonic'' algorithm we have $\l g \r =\l S_{r_j}\r$
and thus $\l g,T\r = \l f,A,T\r$ (with $T = p^{e_1}$ since $e=e_1$).
The ouput Line~\ref{enum:return0} is then in accordance with the specifications.

If this test is not satisfied, more data needs to be output
since additional works must be performed, like Hensel lifting.
We compute  then $S = S_{r_{j+1}}/p^{e_1}$ at
Line~\ref{enum:divbyp} which is not  nilpotent. 
%% We can apply Algorithm~\ref{algo:WeierMonic} ``WeierstrassMonic'' at Line~\ref{enum:WeierA}; It outputs
%% a monic polynomial $g$
%% such that $S_{r_j}= u\cdot g$ where $u$ is a unit in $(k[x]/\l p^{e_1} \r)[y]$.
%% In particular $\l g, p^{e_1}\r =\l S_{r_j} , p^{e_1} \r $.
By the ``last non-nilpotent criterion'' of subresultant~\cite[Thm~5.1]{Da17} 
$\l S_{r_j} ,p^{e_1} \r = \l f,A,p^{e_1}\r$. Hence $\l g, p^{e_1}\r = \l f , A , p^{e_1}\r$, by $(\ddagger)$, as required.

Similarly,  Algorithm~\ref{algo:WeierMonic} at Line~\ref{enum:WeierB} returns a monic polynomial $B$ such that $ S = v\cdot  B$, $v$ being a unit modulo $p^\ell$.
Note that $B$ is monic and satisfies $\deg_y(B) \le \deg_y(g)$,  as required.
\hfill $\Box$ \end{proof}

\begin{Lem}\label{lem:B1}
If a monic polynomial $\ c\not=1\ $ is a common factor of $f$ and $A$ at precision $r > e_1$ then  $\ c\ $ is a monic divisor of $B$ at precision $ r- e_1$.
A  monic divisor $\ c\not=1\ $ of $\ B\ $ is not a common factor of $f$ and $A$
at precision $<e_2$.
\end{Lem}
\begin{proof}
Let $c$ be a factor of $f$ and $A$ at precision $r>e_1$. We can write: $f \equiv c f'  \bmod p^r$
and $A \equiv  c A' \bmod p^r$
(by Definition~\ref{def:precision}, both equalities do not hold together modulo $p^{r+1}$). 
Consider the B\'ezout coefficients (a.k.a cofactors) $u_\iota$, $v_\iota$ of the subresultant $S_{r_\iota}$.
From the equality $S_{r_{j+1}} = u_{j+1} f + v_{_{j+1}}  A$, one deduce that  $S_{r_{j+1}} \equiv c ( u_{j+1} f' + v_{j+1} A') \bmod p^r$.
We also have $p^{e_1} B v \equiv   c ( u_{j+1} f' + v_{j+1} A') \bmod p^r$. Now $c$ being  monic $p^{e_1}$ does not divide $c$,
and we have: $B\equiv c \frac{v^{-1} u_{j+1} f' + v^{-1}  v_{j+1} A'} {p^{e_1}} \bmod p^{r-e_1}$.
Since $r>e_1$ by assumption, $B \equiv 0 \bmod \l p^{r-e_1} ,c\r$ and $c $ is a factor of $B$ at precision at least $r-e_1\ge 1$.
Moreover, since $p^{e_1}$ is the largest power of $p$
that divides $S_{r_{j+1}}$, the precison is exactly $r - e_1$.
This proves the first assertion.

Consider a monic divisor $C$ of $B$ at precision $ < e_2$. This means that the
remainder of the Euclidean division of $B$ by $C$ is in $\l p\r$ but not in $\l p^{e_2}\r$.
If $C$ were a common factor of $A$ and $f$ at precision $\ell < e_2$, 
then it is at precision $\ell \le e_1$ since there is no common factor of $f$
and $A$ at precision $e_1 + 1 ,\ldots, e_2 - 1$, by definition.
We would obtain
$p^{e_1} B   \equiv v^{-1} (u_{j+1} C f' + v_{j+1} C A' + p^{\ell} W) \bmod \l p^e \r$
for a non-zero polynomial $W$, 
and $B \equiv C \frac{ v^{-1} u_{j+1}  f' + v^{-1} v_{j+1}  A'}{p^{e_1}} + p^{\ell-e_1} W \bmod \l p^{e-e_1} \r$.
For $C$ to divide $B$ modulo $p^\ell$, the term $p^{\ell -e_1} W$ shall be zero  at least modulo $p$.
This happens only if $\ell >e_1$ since $W\not=0$. Contradiction with $\ell \le e_1$.
\hfill $\Box$ \end{proof}

\section{The Gcd-Chain Algorithm\label{sec:gcdChain}}

Now that all sub-routines used in Algorithm~\ref{algo:largestFactor} ``largestFactor''
are defined, we introduce in this section main  Algorithm~\ref{algo:gcdChain} below.
It is made of a while loop 
that  has two main components:
\begin{enumerate}
\item a call to Algorithm~\ref{algo:largestFactor} ``largestFactor'' (line~\ref{enum:iterate}) . A loss of precision
is entailed in the division at Line~\ref{enum:divbyp}. %% in the call to Algorithm~\ref{algo:largestFactor}
\item recovering this precision loss at each iteration by Hensel lifting
(Lines~\ref{enum:triggerHL}-\ref{enum:nextblock}).
\end{enumerate}
Each iteration computes one ``block'' of the gcd-chain.
The correctness is shown in Theorem~\ref{thm:gcdChain}, with several preliminaries made of Lemmas~\ref{lem:HL}-\ref{lem:firstblock},
Proposition~\ref{prop:prelim} and its corollary~\ref{cor:prelim}.
These preliminaries contain a detailed description of each step of the algorithm.
\begin{Rmk}\label{eq:exponents}
In the sequel the triangular sets $T_i$, $T'_i$ are
all power of the irreducible polynomial $p$. These exponents
denoted $e_i$ or $e_i'$ 
are introduced only for the analysis
and are not needed in the  algorithms.  
This is consistent with the principle made in \S~\ref{ssec:related}
to not compute squarefree parts.
\end{Rmk}

\begin{algo}
\SetAlgoLined
\SetAlgoNoEnd
\DontPrintSemicolon
\KwIn{Power of an irreducible polynomial  $T=(p(x)^e)$\\
  Univariate polynomials $a$ and $b$ in $R[y]$ where $R = k[x]/\l T \r$}

\KwOut{Lists of polynomials $\calC = [g_1,\ldots,g_s]$  , $\calD=[G_1,\ldots,G_{s-1}]$\\
List of polynomials equal to powers of the irreducible polynomial $p$:  $\Tree = [p^{e_1} , p^{e_2},\ldots,p^{e_s}]$
such that $[(g_1,p^{e_1}),\ldots , (g_s,p^{e_s})]$ is the gcd-chain  of $(a,b,T)$
}

\medskip

$g_0 \leftarrow a $\ ; \  \ $B_0  \leftarrow b$ \ ;\  $S_0 \leftarrow T$\ ;\   $i \leftarrow 0$\ ; \ $e_0 \leftarrow 0$\ ; 
\ \ $T_0 \leftarrow T$\; 

$\calC\leftarrow [\, ]$\ ; \ \ $\calD \leftarrow [\, ]$ \ ; \ \ $\Tree \leftarrow [\, ]$ \;

\While(\nllabel{enum:while}\tcp*[f]{Algo.~\ref{algo:largestFactor} ended
  Line~\ref{enum:enda} or \ref{enum:endb} or \ref{enum:returnzero}}){$B_i \not= $``$end$'' and $B_i\not=1$}{

  $g_{i+1},  \ T_{i+1}'  ,\ B_{i+1} ,\ S_{i+1} \leftarrow  \LargestFactor(g_i,B_i, S_i) $ 
  \hfill \nllabel{enum:iterate}\tcp*[f]{Iteration of Algo.~\ref{algo:largestFactor}} \;

%%     $S_{i+1} \leftarrow S_i / T_{i+1}'$ \hfill \tcp*[f]{$S_{i+1}= p^{e -e_{i+1}} $} \;

 \If(\tcp*[f]{First iteration treated apart}){$i= 0$\nllabel{enum:if0}}{
  $\Tree \leftarrow \Tree {\rm ~cat~}   [     T_{i+1}' ] $\nllabel{enum:42}\; %  \hfill \nllabel{enum:firstblock}\tcp*[f]{First block computed}  \;
}\Else(\nllabel{enum:triggerHL}\tcp*[f]{Trigger Hensel Lifting from second iteration}){

   $G_i \leftarrow  g_i / g_{i+1} \bmod \l T'_{i+1}\r$ \hfill \tcp*[f]{$g_i = G_i  g_{i+1} \bmod \l T'_{i+1}\r $}\nllabel{enum:Gi}\;

   $\alpha,\beta\leftarrow \xgcd( G_i , g_{i+1}  ) \bmod \l T_{i+1}'\r $ \nllabel{enum:xSres}\hfill \tcp*[f]{$\alpha  G_i  + \beta g_{i+1} \equiv 1 \bmod \l T'_{i+1}\r$}\; %%, $u$ unit}\; %  mod. $T'$}\;

 $T_{i+1}\leftarrow T_i  \cdot T'_{i+1} $\nllabel{enum:Tnext} \hfill \tcp*[f]{$T_{i+1}=p^{e_{i+1} }$} \;
 
 ${\rm precision } \leftarrow  \deg(T_{i+1})$\nllabel{enum:precision} \;  

   $ G_i^\star , g_{i+1}^\star \leftarrow {\rm HenselLift}( g_i ,  G_i ,  g_{i+1} , \alpha , \beta ,    T'_{i+1} , {\rm precision})$   \nllabel{enum:HL}\;

   \tcp*[f]{$G_i^\star \equiv G_i \bmod \l T'_{i+1}\r$ ,\ 
     $g_{i+1}^\star \equiv g_{i+1} \bmod \l T'_{i+1}\r  $,
     $g_i \equiv G_i^\star g_{i+1}^\star \bmod \l T_{i+1}\r$}\;
 
   $g_{i+1} \leftarrow g_{i+1}^\star \bmod \l T_{i+1} \r$\nllabel{enum:Anext}\;

   $\calD \leftarrow   \calD {\rm ~ cat~} [ G_i^\star \bmod \l T_i \r]$
   \nllabel{enum:block} \; %%  \ ;\ \ $\Tree \leftarrow \Tree {\rm ~cat~}   [ T' \cdot T_{prev}] $\; %% \hfill \nllabel{enum:block}\tcp*[f]{Add the block $( , TO \ DO)$ to }  \;

  $\Tree \leftarrow \Tree {\rm ~cat~}   [     T_{i+1} ] $  \hfill \nllabel{enum:nextblock}\tcp*[f]{Next block computed}  \;
}

  $\calC \leftarrow   \calC {\rm ~ cat~} [g_{i+1} ]$\nllabel{enum:52}\;

 $i\leftarrow i+1$\; 
}
\Return{$\calC$, $\calD$, $\Tree$}
\caption{\label{algo:gcdChain} The gcd-chain algorithm}
\end{algo}

%% The analysis starts with that of the Hensel lifting.

\begin{Lem}\label{lem:HL}
The Hensel lifting at Line~\ref{enum:HL}
of Algorithm~\ref{algo:gcdChain} returns
$G_i^\star, \ g_{i+1}^\star$ verifying:
$G_i^\star \equiv G_i \bmod \l T'_{i+1}\r$, 
$g_{i+1}^\star \equiv g_{i+1} \bmod \l T'_{i+1}\r  $,
$g_i \equiv G_i^\star g_{i+1}^\star \bmod \l T_{i+1}\r$.
%% It requires $O(\M(\deg_y(g_i)) \M({\rm precision}))$ operations in $k$.
\end{Lem}

\begin{proof}
This is classical: the algorithm follows exactly the steps 
presented in Algorithm~15.10 of ~\cite{GaGe03}.
%% 
%% The complexity estimate follows from Theorem~15.11 of~\cite{GaGe03}.
%% The number of iterations required is
%% $\lceil \log_2(\log_{d'}({\rm precision})) \rceil$ where $d'=\deg_x(T'_{i+1})$
%% and ${\rm precision}$ is defined
%% at Line~\ref{enum:precision}. 
\hfill $\Box$ \end{proof}

The next lemma clarifies the status of the input/output when calling Algorithm ``largestFactor'' at Line~\ref{enum:iterate},
within the several iterations. The notations introduced  will be used thereafter.

\begin{Lem}\label{lem:in-out}
 Write $\Tree = [T_1,\ldots,T_s]$ with $T_i = p^{e_i}$, 
 $\calC=[g_1,\ldots,g_s]$ 
and $\calD = [G_1,\ldots,G_{s-1}]$
the three outputs of 
Algorithm~\ref{algo:gcdChain}.
For all $i=0,\ldots,s-1$, consider the output
$g_{i+1},\ T_{i+1}', \ B_{i+1},\ S_{i+1} $ of the call to ``largestFactor''
at Line~\ref{enum:iterate}, and $T_{i+1}$ as defined at Line~\ref{enum:Tnext}. They are related as follows:
$$
T_{i+1}' = p^{e_{i+1} - e_i},\qquad  S_{i+1} = p^{e-e_{i+1}}, \qquad T_{i+1} = p^{e_{i+1}} \qquad ({\rm with~} e_0:=0). 
$$
\end{Lem}
\begin{proof}
The proof goes by induction on $i=0,\ldots,s-1$. 
According to the specifications of Algorithm~\ref{algo:largestFactor}
``largestFactor'' (Line~\ref{enum:iterate}) proved in Lemma~\ref{lem:correct}.
we have when $i=0$, $T'_1 = p^{e_1}$, $S_1=p^{e-e_1}$. Note that the case $i=0$
(Line~\ref{enum:if0}) does not require Hensel lifting, and thus $T_1 = T_{1}' = p^{e_1-e_0} = p^{e_1}$,
as required.

By induction hypothesis we can assume that the result holds up to $i < s-1$.
There is an $i+1$-th call to Algorithm~\ref{algo:largestFactor}
``largestFactor'' (at Line~\ref{enum:iterate}), with input $S_{i} = p^{e-e_i}$, by induction hypothesis.
From the Specification of Algorithm~\ref{algo:largestFactor},
the output is $T'_{i+1} = p^{e_{i+1}- e_i}$ because the input $S_i$ is only at precision
$e - e_i$;
And $S_{i+1} = p^{e-e_i - (e_{i+1}-e_i)} = p^{e - e_{i+1}}$ as required.
The definition of $T_{i+1}$ at Line~\ref{enum:Tnext} gives
$T_{i+1} = T_{i+1}' \cdot T_i = p^{e_{i+1}-e_{i}} p^{e_i} = p^{e_{i+1}}$ as required.
\hfill $\Box$ \end{proof}

The proposition below connects the different outputs obtained after each iteration,
to the initial input $a,b,T$. It is crucial in the proof of Theorem~\ref{thm:gcdChain}

\begin{proposition}\label{prop:prelim}
With the notation of Lemma~\ref{lem:in-out},
we have
$$
\text{for~} i\ge 1,\qquad   p^{e_{i-1}} g_i \in \l a,\ b,\  p^e \r, \qquad p^{e_{i}} B_i \in \l a,\ b,\ p^e\r.
$$
\end{proposition}
\begin{proof}
We proceed by induction on $i$, starting with the base case $i=1$.

By definition $g_1$ is the first output polynomial of the first call to Algorithm~\ref{algo:largestFactor}
``$\LargestFactor$'' in Algorithm~\ref{algo:gcdChain} ``gcdChain''.
From the definition of the Weierstrass factorization at Line~\ref{enum:WeierA}
of Algorithm~\ref{algo:largestFactor}, there is a unit $\nu_1 \in ( k[x]/\l p^e\r )[y]$ such that
$\nu _1 g_1 = S_{r_j}$. Write the B\'ezout coefficients $u^{(1)}_j a + v^{(1)}_j b = S_{r_j}$,
so that $g_1 = \nu_1^{-1} (  u^{(1)}_j a + v^{(1)}_j b)$; this proves
that $g_1 \in \l a,\ b,\ T\r$ (and that $g_1 \in \l a,\ b,\ p^e\r$: indeed $p^e=T$).

$B_1$ is the third  output polynomial of the first call to Algorithm~\ref{algo:largestFactor}
``$\LargestFactor$'' in Algorithm~\ref{algo:gcdChain} ``gcdChain''.
From Line~\ref{enum:WeierB} of Algorithm~\ref{algo:largestFactor} and by definition of the Weierstrass factorization,
there is a unit $\epsilon_1 \in (k[x]/\l p^{e-e_1}\r)[y]$ such that $p^{e_1} \epsilon_1 B_1 = S_{r_{j+1}}$.
With the B\'ezout coefficients written $S_{r_{j+1}} = u^{(1)}_{j+1} a + v^{(1)}_{j+1} b$, we obtain
$p^{e_1} B_1 = \epsilon_1^{-1}(u^{(1)}_{j+1} a + v^{(1)}_{j+1} b)\in 
\l a, b, p^e\r$. 

Next assume that the result holds up to a value $1\le i<s$ and let us prove it for $i+1$.
By definition $g_i$ (resp. $B_i$) is the first (resp. the third) output polynomial of the $i$-th call to Algorithm~\ref{algo:largestFactor}
``$\LargestFactor$'' in Algorithm~\ref{algo:gcdChain} ``gcdChain''.
From the definition of the Weierstrass factorization at
Line~\ref{enum:WeierA} (resp. at Line~\ref{enum:WeierB})
of Algorithm~\ref{algo:largestFactor}, there is a unit $\nu_{i+1} \in ( k[x]/\l p^{e-e_i}\r )[y]$  (resp. $\epsilon_{i+1}$) such that
$ g_{i+1} \nu_{i+1}= S_{r_j} $ (resp. $p^{e_{i+1}-e_i} \epsilon_{i+1} B_{i+1} = S_{r_{j+1}}$).
Write the B\'ezout coefficients as follows:
$$
S_{r_j} = u_{j}^{(i+1)} g_i + v_j^{(i+1)} B_i,\quad S_{r_{j+1}} = u_{j+1}^{(i+1)} g_i + v_{j+1}^{(i+1)} B_i. 
$$
It follows that: 
$$
p^{e_i} g_{i+1}  \stackrel{\small (\bullet)}{=} \nu_{i+1}^{-1} ( u_{j}^{(i+1)} p^{e_i} g_i + v_j^{(i+1)} p^{e_i} B_i) ,\qquad 
p^{e_{i+1} - e_i} B_{i+1} = \epsilon_{i+1}^{-1} (u_{j+1}^{(i+1)} g_i + v_{j+1}^{(i+1)} B_i). 
$$ 
The latter equality implies $p^{e_{i+1}} B_{i+1} \stackrel{\small (\star)}{=} \epsilon_{i+1}^{-1} (u_{j+1}^{(i+1)} p^{e_i} g_i + v_{j+1}^{(i+1)} p^{e_i} B_i)$.
By induction hypothesis $p^{e_{i-1}} g_i \in \l a,b,p^e\r$ and $p^{e_i} B_i \in \l a,b,p^e\r$.
Thus $p^{e_i} g_i = p^{e_i - e_{i-1}} p ^{e_{i-1}} g_i \in \l a , b, p^e\r $. 
Consequently, both  the r.h.s of Eqs.~$(\star)$ and $(\bullet)$
are in $\l a,b,p^e\r$, 
therefore so are $p^{e_i} g_{i+1} $ and $p^{e_{i+1}} B_{i+1}$. 
\hfill $\Box$ \end{proof}

The first iteration of the loop of Line~\ref{enum:while} is special
since it doesn't require Hensel Lifting (case $i=0$). The lemma hereunder
treats this base case apart, and used  in the induction proof of Theorem~\ref{thm:gcdChain}.

\begin{Lem}\label{lem:firstblock}
The first iteration of the while loop (Line~\ref{enum:while}, $i=0$)
of Algorithm~\ref{algo:gcdChain}
fills $\calC$ with $g_1$ and $\Tree$ with $T_1'=p^{e_1}$.
We have $\l g_1,p^{e_1}\r = \l a, b,p^{e_1} \r$.

Moreover, if there is only one iteration the stronger equality
 $\l g_1 , p^{e_1}\r = \l a, b , p^{e} \r$ holds.
\end{Lem}
\begin{proof}
From the proof of correctness of Algorithm~\ref{algo:largestFactor}
one has at Line~\ref{enum:iterate} of Algorithm~\ref{algo:gcdChain}
$\l g_1, T'_1\r = \l g_0 ,B_0 ,T_1'\r$.
But $T_1'=p^{e_1}$, $g_0=a$ and $B_0=b$  yielding the first equality.

For the second one,
by Proposition~\ref{prop:prelim} we know that $p^{e_0} g_1=g_1\in \l a, b, p^e\r$ 
and that $p^{e_1} B_1 \in \l a,b,p^{e}\r$.
If there is only one iteration
of the while loop Line~\ref{enum:while},
namely for $i=0$, then there are two sub-cases.
First $B_1\not=$``$end$'' and $B_1=1$.
Then Proposition~\ref{prop:prelim} gives $p^{e_1} B_1 = p^{e_1} \in \l a,b,p^{e}\r$
which proves $\l p^{e_1} , g_1\r  = \l a,b,p^e\r$.
Second, $B_1=$``$end$''.
This happens when the call
to Algorithm ``largestFactor'' (Line~\ref{enum:iterate})
with input $a,b,T$
outputs $B_{i+1}=B_1=$``$end$''.
According to Lines~\ref{enum:enda},\ \ref{enum:endb},\ \ref{enum:returnzero}
the output denoted $p^{e_1}$ (in this Lemma)
is equal to $T$ (as denoted in Algorithm~\ref{algo:largestFactor})
which is equal to $p^e$.
This implies $\l g_1 ,p^{e_1}\r=\l g_1,p^e\r=\l a,b,p^e\r$.
\hfill $\Box$ \end{proof}

To prove correctness of the ``gcdChain'' Algorithm~\ref{algo:gcdChain},
one must consider the isomorphisms~\eqref{eq:isochain} and~\eqref{eq:block}.
This amounts to prove that a product of ideals of type
$\l G_1,p^{e_1}\r \l G_2, p^{e_2}\r \cdots \l g_i,p^{e_i}\r$ is equal 
to ideals of type $\l a,b,p^{e_i}\r$. Thus we must examine the generators of the
product of ideals, which is the purpose of the following corollary (of Proposition~\ref{prop:prelim}).

\begin{corollary}\label{cor:prelim}
The same notations as in Lemma~\ref{lem:in-out} are used.
Let $\calI$ be a subset of indices in $\jump{1,s-1}$ and let $\bar{\calI}$ its complement
in $\jump{1,s}$ (so that $\calI,\bar{\calI}$ is a partition $\jump{1,s}$).

Then $g_s \prod_{i\in \calI} p^{e_i} \prod_{j\in \bar{\calI}} G_j \in \l a,b,p^e\r$. 
\end{corollary}
\begin{proof}
Assume that $\calI\not=\emptyset$ first and let $\varsigma:=\max \calI$.
Then $\{ \varsigma+1,\ldots,s-1\} \subset \bar{\calI}$ (possibly empty). 
We claim that $p^{e_{\varsigma}} g_s \prod_{j=\varsigma+1}^{s-1} G_j \in \l a,b,p^e\r$, from which  the
conclusion follows since $\l  g_s \prod_{i\in \calI} p^{e_i} \prod_{j\in \bar{\calI}} G_j   \r \subset \l p^{e_{\varsigma}} g_s \prod_{j=\varsigma+1}^{s-1} G_j\r$.

From the definition of $G_i$ and $g_i$ after Hensel Lifting at Lines~\ref{enum:block} and~\ref{enum:Anext}
of Algorithm~\ref{algo:gcdChain} one has: $g_{i-1} = G_{i-1} g_i \bmod p^{e_{i-1}}$ (see Lemma~\ref{lem:HL}).
It implies the following equality where $(\cdots)$ denotes a polynomial:
$$
\begin{array}{rcl}
\left(\prod_{j=\varsigma+1}^{s-1} G_j\right) g_s   & = &  (\prod_{j=\varsigma+1}^{s-2} G_j) ( g_{s-1} + p^{e_{s-1}} g_s (\cdots) )
=   (\prod_{j=\varsigma+1}^{s-3} G_j) ( g_{s-2} +  p^{e_{s-1}} g_{s-1} (\cdots)  +   p^{e_{s-1}} g_s ( \cdots))
\\
& = &   g_{\varsigma+1} + p^{e_{\varsigma+1}} g_{\varsigma+2} (\cdots) + 
 \cdots +  p^{e_{s-3}} g_{s-2} (\cdots)   +  p^{e_{s-2}} g_{s-1} (\cdots)  +   p^{e_{s-1}} g_s (\cdots))
\end{array}
$$
If we multiply the above equation by $p^{e_\varsigma}$ then
each term is a multiple of $p^{e_\ell} g_{\ell+1}$ for some $\ell=\varsigma,\ldots,s-1$.
Hence, by Proposition~\ref{prop:prelim} each term is in $\l a , b ,p^e\r$ yielding the conclusion in the case $\calI\not=\emptyset$.

Otherwise when $\calI= \emptyset $, the product is equal to $G_1 \cdots G_{s-1} g_s$ and similarly to the above, we get 
$g_1 + p^{e_1} g_2 (\cdots)  + \cdots + p^{e_{s-1}} g_s (\cdots)$, which is in $\l a, b, p^e\r$
according to Proposition~\ref{prop:prelim}.
\hfill $\Box$ \end{proof}

The proof of Theorem~\ref{thm:gcdChain} requires a last intermediate result.
\begin{Lem}\label{lem:inclusion}
With the notations of Lemma~\ref{lem:in-out}
assume that $s\ge 2$. Fix an  $ 1\le i \le s-1$ and refer to  the notations of Eqs~\eqref{eq:a=Gia}-\eqref{eq:Gia},
to define $a' : =  G_{i+1}^{(a)} \cdots G_s^{(a)} \cdot a_{\gamma+1}\cdots a_\alpha $ 
and
$b' : =  G_{i+1}^{(b)} \cdots G_s^{(b)} \cdot b_{\gamma+1}\cdots b_\beta $.
We have:
\begin{enumerate}
\item\label{enum:a'} $\l a , b , p^{e_{i+1}}\r \subset \l G_1 , p^{e_1}\r \cdots \l G_i  , p^{e_i}\r \l a',b', p^{e_{i+1}}\r$.
\item\label{enum:a} $\l a' , b' , p^{e_{i+1}}\r = \l g_{i+1}, p^{e_{i+1}}\r$.
\end{enumerate}
\end{Lem}
\begin{proof}
By Eqs.\eqref{eq:a=Gia}-\eqref{eq:Gia},
we have $a = a' \prod_{j=1}^i G_j^{(a)} $ (similarly 
$b = b' \prod_{j=1}^i G_j^{(b)} $)
with $G_j^{(a)}\equiv c_j^{(a)} G_j \bmod \l p^{e_j}\r$
(respectively $G_j^{(b)}\equiv c_j^{(b)} G_j \bmod \l p^{e_j}\r$).
Therefore
$$
a = a' \cdot \prod_{j=1}^i ( c_j^{(a)} G_j + p^{e_j} (\cdots) ) ,
\quad 
b = b' \cdot \prod_{j=1}^i ( c_j^{(b)} G_j + p^{e_j} (\cdots) ) ,
$$
where $(\cdots)$ denotes some polynomials.
This proves the inclusion~\ref{enum:a'}.

To prove~\ref{enum:a}, let us treat first the case $i=1$.
Consider the polynomials  $g_1 , B_1$, which are the output of the first call
to ``largestFactor'' at Line~\ref{enum:iterate} of Algorithm~\ref{algo:gcdChain}.
Note that if we apply the lemma~\ref{lem:firstblock} with $g_1 , B_1 , p^{e-e_1}$
instead of $a,b,T$, we obtain the ideal equality: $\l g_1 ,B_1, p^{e_2 - e_1}\r =
\l g_2, p^{e_2-e_1}\r$; Note the precision loss of $e_1$  entailed at the previous iteration $i=0$.
Let us prove  
 prove that $\l a' ,b' , p^{e_2-e_1}\r = \l g_1, B_1,p^{e_2-e_1}\r$.
According to Definition~\ref{def:precision} and to the definitions of the polynomials
$a'$ and $b'$,
the common factors and $a'$ and $b'$ are exactly those of $a$ and $b$ at precision $\ge e_2$.
By Lemma~\ref{lem:B1}, such common factors are divisors
of $B_1$ at precision $\ge e_2 - e_1$. 
On the other hand, by Lemma~\ref{lem:correct}
the divisors of $g_1$ (at precision $p$) are {\em exactly}
the common factors of $a$ and $b$.
Therefore, {\em up to precision},   the common factors of $g_1$ and $B_1$ are {\em exactly} 
those of $a$ and $b$ at precision $\ge e_2$, hence
are the same as those of $a'$ and $b'$.
As for {\em precision},
the common factors at precision $\ge e_2$ of $a,b$,
and those of $g_1,B_1$ (and they are at precision $\ge e_2 -e_1$ due to precision loss).
We deduce that $\l g_2  , p^{e_2- e_1}\r = \l a',b' ,p^{e_2-e_1}\r$.
After Hensel lifting at Line~\ref{enum:HL}, we obtain $\l g_2  , p^{e_2}\r = \l a',b' ,p^{e_2}\r$ .

The general case $i>1$ reduces to the case $i=1$ by considering instead
of $a$ and $b$,  $E := a/G_1^{(a)}\cdots G_{i-1}^{(a)} = a_{\gamma+1}\cdots a_\alpha\cdot \prod_{j=i}^s G_{j}^{(a)}$
and respectively $F := b/G_1^{(b)} \cdots G_{i-1}^{(a)} = b_{\gamma+1}\cdots b_\beta\cdot \prod_{j=i}^s G_{j}^{(b)}$;
And instead of $a'$ and  $b'$ to take
$E':=a' / G_1^{(a)} \cdots G_{i}^{(a)} = a_{\gamma+1}\cdots a_\alpha\cdot \prod_{j=i+1}^s G_{j}^{(a)}$
and respectively $F' := b' / G_1^{(b)}  \cdots G_{i}^{(b)} = b_{\gamma+1}\cdots b_\beta\cdot \prod_{j=i+1}^s G_{j}^{(b)}$.
\hfill $\Box$ \end{proof}

\begin{theorem}\label{thm:gcdChain}
At the end of the  $i$-th iteration
of the while loop (Line~\ref{enum:while}) of Algorithm~\ref{algo:gcdChain},
the output lists $\calC, \Tree, \calD$ are currently filled with:
$$\calC=[ g_1,\ldots,g_i], \quad
\Tree=[p^{e_1},\ldots,p^{e_i}], \quad 
\calD=[G_1,\ldots,G_{i-1}],
$$
and at the end of the algorithm, after say $s$ iterations,
the sequence $[(g_1,p^{e_1}),\ \ldots,\ (g_s,p^{e_s})]$ is a gcd chain.
\end{theorem}
\begin{proof}
We must show that Isomorphisms~\eqref{eq:isochain} and~\eqref{eq:block} hold.
By the Chinese remaindering theorem it suffices to show that
\begin{equation}\label{eq:isochain2}
 \l a, b , p^{e}\r = \l G_1 , p^{e_1} \r \cdots \l G_{s-1} , p^{e_{s-1}}\r \l g_s , p^{e_s}\r.
\end{equation}
holds, in order to prove~\eqref{eq:isochain}; and to prove that~\eqref{eq:block} holds, it suffices to show that:
\begin{equation}\label{eq:block2}
\text{for all } s \ge  i \ge 1,\ \ \ 
\l a, b , p^{e_i}\r = \l G_1 , p^{e_1} \r \cdots \l G_{i-1} , p^{e_{i-1}}\r \l g_i , p^{e_i}\r,
\end{equation}

Assume first that $s = 1$. 
Then Eq.~\eqref{eq:block2} and Eq.~\eqref{eq:isochain2} corresponds
respectively to the first and to the second equality of Lemma~\ref{lem:firstblock}.

Assume now that $s > 1$. Let us prove Eq.~\eqref{eq:block2} first.
The   case $i=1$ amounts to
$\l g_1 , p^{e_1} \r = \l a , b , p^{e_1}\r$
and is provided by Lemma~\ref{lem:firstblock}.
Assuming now $i>1$.
Statements~\ref{enum:a'}.-\ref{enum:a} of Lemma~\ref{lem:inclusion} together insure that
$\l a,b,p^{e_i}\r \subset \l G_1 ,p^{e_1}\r \cdots \l G_{i-1}, p^{e_{i-1}}\r \l g_i ,p^{e_i}\r$ ($\dagger$).

On the other hand, Corollary~\ref{cor:prelim} insures of the inclusion:
$$
\l G_1 , p^{e_1} \r \cdots \l G_{i-1} , p^{e_{i-1}}\r \l g_{i} \r \subset \l a, b, p^e   \r.
$$
Indeed the generators of the product ideals on the l.h.s. are proved in that
corollary to belong to the r.h.s. It follows that:
$$
\l G_1 , p^{e_1} \r \cdots \l G_{i-1} , p^{e_{i-1}}\r  
\l g_{i},p^{e_{i}}\r  \subset \l a, b, p^{e_{i}}   \r.
$$
Together with $(\dagger)$, this proves Eq.~\eqref{eq:block2}.
Let us prove  Eq.~\eqref{eq:isochain2}.  Eq.~\eqref{eq:block2} with $i=s$ yields:
$$
\l G_1 , p^{e_1} \r \cdots \l G_{s-1} , p^{e_{s-1}}\r \l g_{s} , p^{e_{s}}\r =   \l a, b, p^{e_s}   \r.
$$
Corollary~\ref{cor:prelim} provides the inclusion:
$$
 \l a, b , p^{e}\r \supset \l G_1 , p^{e_1} \r \cdots \l G_{s-1} , p^{e_{s-1}}\r \l g_s \r.
$$
Moreover, since by assumption there are  $s$ iterations, the while loop at Line~\ref{enum:while}
of Algorithm~\ref{algo:gcdChain} stops at the $s+1$-th iteration.
Either  $B_s\not=$``$end$'' 
and then $B_s=1$.
Proposition~\ref{prop:prelim}  then
guarantees that $p^{e_s} B_s = p^{e_s}\in \l a,b,p^{e}\r$. We obtain the equality~\eqref{eq:isochain2}
required.
Either  $B_s=$``$end$'': the $s-1$-th call to ``largestFactor'' at Line~\ref{enum:iterate}
with input $g_{s-1}, B_{s-1},p^{e-e_{s-1}}$  exits at one of the Lines~\ref{enum:enda},~\ref{enum:endb},~\ref{enum:returnzero}.
If it is at Line~\ref{enum:enda}, then $B_{s-1}=0$ which happens eventually only if $b=0$, that is if there is only one iteration;
 Since $s>1$ here, this cannot happen.
%% Otherwise, this  case is treated at the {\em  previous} iteration at Line~\ref{enum:returnzero}.
If it is at Line~\ref{enum:endb}, then $B_{s-1}=1$, and there is no such $s$-th iteration since
the test at Line~\ref{enum:while} is not passed. It remains the case of Line~\ref{enum:returnzero}.
Then the output  $T$  is equal here to $p^{e_s}$ and we have $e=e_s$. Then Eq.~\eqref{eq:isochain2}
and Eq.~\eqref{eq:block2} coincide.

\hfill $\Box$ \end{proof}

%% 
%% 
%% 
%% and after the Hensel lifting of Line~\ref{enum:HL}
%% $g_i \equiv G^\star A^\star \bmod \l T_{next}\r$.
%% Let $g_{i+1} := A_{next}$ (Line~\ref{enum:Anext}), $p^{e_{i+1}} := T_{next}$ (Line~\ref{enum:Tnext}),
%% and $G_{i}:=G^\star \bmod \l p^{e_{i}}\r$ (Line~\ref{enum:block}). 
%% %% According to the specification of largestFactor, holds $\l g_{i+1} , p^{e_{i+1}}\r = \l g_i , B_{next}, p^{e_{i+1}}\r$.
%% We have $\l g_{i+1} \ , \ p^{e_{i}} \r \cdot \l G_i , p^{e_i}\r = \l g_i, p^{e_i}\r$ 
%% %% and  $\l g_{i+1} \ , \ p^{e_{i+1}} \r \cdot \l G^\star , p^{e_{i+1}}\r = \l g_i, p^{e_{i+1}}\r$.
%% On the other hand by induction hypothesis, 
%% $\l a, b , p^{e_i}\r = \l G_1,p^{e_1} \r \cdots \l G_{i-1} , p^{e_{i-1}}\r \l g_i , p^{e_i}\r$.
%% Therefore $\l a, b , p^{e_i}\r = \l G_1,p^{e_1} \r \cdots \l G_{i-1} , p^{e_{i-1}}\r\l G_{i} , p^{e_{i}}\r \l g_{i+1} , p^{e_i}\r$.
%% It is moreover a direct product. 
%% Now $\l a, b , p^{e_i}\r\supsetneq \l a, b , p^{e_{i+1}}\r $. Now $p^{e_{i+1}}$ is contained in all factors in the r.h.s.,
%% 
%% TO DO
%% 
%% 
%% $\l g_{i+1} \ , \ p^{e_{i+1}} \r +  \l G_i , p^{e_i}\r= \l 1\r$.
%% 
%% \hfill $\Box$ \end{proof}

\section{A variant of  Weierstrass preparation's theorem\label{sec:Weier}}
The Weierstrass preparation theorem  states that
a formal power series $f = \sum_i a_i X^i \in \fa [[X]]$ with coefficients in a
local complete ring $(\fa,\fm)$, not all of them lying in $\fm$,
has a unique factorization $f = q u$ where $q = q_0 + \cdots + q_{n-1} X^{n-1} + X^n$ is
monic and $q_i \in \fm$,
and where $u \in \fa[[x]]^\star$ is an invertible power series.

In our context the local complete ring is $ \fa = k[x]/\l p^e\r$, $\fm=\l p\r$
(indeed it is equal to $k[[x]]/\l p^e\r$, which is a finite quotient of the local  complete ring
$k[[x]]$).
But the factorization supplied by the
classical version ({\em e.g.}~\cite[Theorem~9.1]{lang2002algebra})
does not fit the needs of this work. The following variant does:

\begin{proposition}\label{prop:WeierFac}
Let $(\fa,\fm)$ be a complete local ring and let $f\in \fa [X]$ a {\em polynomial}, say $f=f_d X^d + \cdots +f_0$ which has not all of its coefficients $f_i$ lying in  $\fm$.
Write $f_k$ the coefficient of highest degree that is not in $\fm$. There are unique polynomials $q$ and $u$ in $\fa[X]$ such that $f = u q $ where:
\begin{itemize}
\item $q$ is monic of degree $k$.
\item $u = u_0 + u_1 X + \cdots + u _{d-k}  X^{d-k}$ where $u_0\not \in \fm$ and $u_i \in \fm$ for $i\ge 1$ (note that  $u$ is a unit of $\fa[[X]]$).
\end{itemize}
\end{proposition}

The proof is adapted from~\cite[Theorem~9.1-9.2]{lang2002algebra}.

\begin{proof}
Write $rev_d(f):=X^d f(\frac 1 X)$ the reversal polynomial of $f$ (this is why we need $f$ to be a polynomial and not a general
power series). The term of smallest degree not in $\fm$ is $f_k X^{d-k}$. By the standard Weierstrass preparation Theorem~9.2 of~\cite{lang2002algebra}
$rev_d(f) = g s$ where $s$ is an invertible power series  in $\fa [[X]]$
and $g = X^{d-k} + g_{d-k-1} X^{d-k-1} +\cdots + g_0$ is a monic polynomial which satisfies
$g_i\in \fm$. Since $\deg(rev_d(f))= d$ and that $g$ is monic, necessarily $s\in \fa[X]$ and is of degree $k$.
We write $s=s_0 + \cdots + s_k X^k$. Thus:
$$
rev_d(rev_d(f)) = f = rev_{d-k} (g) rev_k (s) = (1 + g_{d-k-1} X + \cdots + g_0 X^{d-k} ) ( s_k + \cdots +  s_0 X^k )
$$
Now since $s$ is invertible, $s_0\in \fa^\star$.
Letting $ u = s_0  + \cdots + g_0 s_0 X^{d-k}$ and $ q = s_0^{-1} s_k +\cdots + s_0^{-1} s_1 X^{k -1}+ X^k$ provides
polynomials satisfying  $f = u q$ as well as the requirement of the proposition.
%% 
%%  $f= (s_0  + \cdots + g_0 s_0 X^{d-k}) ( s_0^{-1} s_k +\cdots + X^k)  = u q$. Indeed
%% $u = s_0 +\cdots$ is a unit in $\fm [[X]]$ with $s_i \in \fm$ for $i\ge 1$; And $q$ is monic.
\hfill $\Box$ \end{proof}

% In our algorithm, only the polynomial $q$ is necessary. 
The polynomial  $q$ can be computed by an Euclidean division.
This possibility is a minor adaptation of~\cite[Theorem~9.1]{lang2002algebra},
that we have reproduced in Appendix for sake of completeness.

\begin{corollary}\label{cor:div}
With the same notations and hypotheses of Proposition~\ref{prop:WeierFac},
given the inverse coefficient $f_k^{-1}$, the monic polynomial
$q$ can be computed by mimicking the Euclidean
division of $X^k$ by $f_k X^k + f_{k-1} X^{k-1} +\cdots$,
but with $f$ instead of $f_k X^k + f_{k-1} X^{k-1} +\cdots$
to obtain $X^k =  g f + r$.
The power series $g$ 
needs to be computed only up to modulo $X^{d-k+1}$.
Followed
by the inversion of the  truncated power  series $g$ to get: $q = (X^k - r) g^{-1}$.
%%
%%If $\fa = k[x]/\l p^e\r$, then this requires up to $O(\M( k ))$ operations in $R$.
\end{corollary}

\begin{proof}
By Lemma~\ref{lem:WeierDiv}, there exits a unique invertible power series 
$g\in \fa[[X]]$ and a unique polynomial $r\in \fa[X]$ such that 
$X^k =  g  f + r$.  Therefore $(X^k-r) g^{-1} = f$ is the factorization of Proposition~\ref{prop:WeierFac},
by uniqueness. We deduce that
$g^{-1}$ is a polynomial of degree $d-k$ that verifies the conditions of the 
aforementioned proposition.
In consequence, the power series $g$ needs to be known only modulo $X^{d-k+1}$.
%% 
%% The cost analysis of the Euclidean division
%% is always lower than $\M(d)$ (\cite[Thm~9.6]{GaGe03})
%% and the inversion of $g$ is less than $\M(d-k+1)$
%% 
\hfill $\Box$ \end{proof}
 
\begin{ex}\label{ex:1}
Let $\fa=k[[x]]/\l x^2\r \simeq k[x]/\l x^2\r$ be the complete local ring of maximal ideal
$\fm=\l x\r$, $f=x y^2 + y + 1$. The Weierstrass factorization of Proposition~\ref{prop:WeierFac}
insures the existence of a polynomial  $q = y + \cdots\in \fa[y]$ and $u = u_0 + u_1 y $ with $u_0\in \fa^\star$
and $u_1\in \fm$, such that $f= u q$. To compute it,
Euclidean division as explained in Corollary~\ref{cor:div} works as follows (boxed terms are parts of the remainder).
$$
\begin{array}{ccccr|r}
   & y & & & &   x y^2 + y + 1 \\
  \cline{6-6}
-x y^2 & &  & & \fbox{-1} & 1 - x y + x  \\
& x y & & & &  \\
& & & & \fbox{-x} & 
\end{array}
$$
We obtain $y = ( x y^2 + y + 1)(1 - x y + x )  -x -1 $, hence
$ q u =(y+ x +1) (1-x y +x)^{-1}=x y^2 + y + 1$ is the factorization of Proposition~\ref{prop:WeierFac}.
Now $u=(1-x y +x)^{-1} = 1 + \sum_{i\ge 1} x^i (y-1)^i = 1+ x (y-1)$ in $(k[x]/\l x^2\r )[y]$. We check that $u_0= 1-x \in \fa^\star$ and $u_1=x \in \fm$.
\end{ex}

\section{Concluding remarks}
%% \subsection{Towards triangular sets of several variables}

\subsubsection*{A running example} 
Let $p(x)=x$, $e=7$, that is the input primary triangular set in one variable is $T(x)=x^7$.
The two input monic polynomials in $(\Q[x]/\l T(x)\r)[y]$ are defined
through their unique factorization into (two) coprime factors as:
$$
\begin{array}{rcl}
a & := & (y+x+x^2+x^3+x^4+x^5+x^6)*(y+1+x^3+x^6)\\
b & := & (y+x+x^2+2*x^3+x^4+x^5+x^6)*(y+1+x^3+x^5)
\end{array}
$$

\noindent gcdChain$(a,b,x^7)==$
\begin{itemize}

\item First call to ``largestFactor'' Line~\ref{enum:while} ($i=0$): $\LargestFactor (a,b,x^7) ==$
  \begin{itemize}
    \item Line~\ref{enum:28}     $\subres (a,b,x^7) ==$
    \begin{itemize}
     \item Line~\ref{enum:sresT}:
       Subresultant mod $x^7$ does not fail; We obtain $\bfS=[a,b,\ldots]$, precisely:
       $$
       \begin{array}{l}
         \bfS = [
           y^2 + (2*x^6 + x^5 + x^4 + 2*x^3 + x^2 + x + 1)*y + 2*x^6 + 2*x^5 + 2*x^4 + x^3 + x^2 + x,\\
           y^2 + (x^6 + 2*x^5 + x^4 + 3*x^3 + x^2 + x + 1)*y + 4*x^6 + 2*x^5 + 2*x^4 + 2*x^3 + x^2 + x,\\
           (-x^6 + x^5 + x^3)*y + 2*x^6 + x^3,\\
           0
         ]
       \end{array}
       $$
       \item The last non nilpotent subresultant is $b=\bfS[2]$ and the first nilpotent 
       is $\bfS[3] = (-x^6 + x^5 + x^3)*y + 2*x^6 + x^3$.

       \item While loop Line~\ref{enum:stillNil} ($j=3$) $S_{r_j} = \bfS[3]$ \\
         $\nilpotentFactor(S_{r_{j}} , x^7)== x^3, -1$  $\qquad (P_{new}=x^3$,  
         $x^3\mid \bfS[3],\ i_{new}=-1)$
       \item While loop Line~\ref{enum:stillNil} ($j=2$) $S_{r_j} = \bfS[2]=b$ \\
         $\nilpotentFactor(S_{r_{j}} , x^7)== 1, 2$ $\qquad (P_{new}=1$,
         $\coeff(b,i_{new})=\coeff(b,y^2)=1)$
       \item exit while loop and return $\subres (a,b,T) == \bfS, 2, 2, x^3$
    \end{itemize}
  \item Line~\ref{enum:WeierA}: $\WeierstrassMonic(b, x^7, 2) == b$
  \item Line~\ref{enum:divbyp}: $S = \bfS[3]/x^3 =  (-x^3 + x^2+1)*y + 2*x^3 +1 $
  \item Lines~\ref{enum:33}-\ref{enum:34}: $p^\ell = T/x^3 = x^4$, 
    $  \qquad \nilpotentFactor(S, x^4)== 1, 1$ ($i=1$) 
  \item Line~\ref{enum:WeierB}: $B_1 = y + 3 *x^3 - x^2 + 1 $
  \item return $b, x^3, B_1, x^4$
  \end{itemize}
  \item Line~\ref{enum:42}, Line~\ref{enum:52}: $\Tree=[x^3]$, $\calC=[g_1]$ ($g_1=b \bmod \l x^3\r$)
  \item Line~\ref{enum:while}  ($i=1$): $\LargestFactor(g_1, B_1 , x^4) ==$
    \begin{itemize}
      \item Line~\ref{enum:28}:     $\subres (g_1,B_1 ,x^4) ==$
      \begin{itemize}
      \item Line~\ref{enum:sresT} Subresultant mod $x^4$ does not fail; We obtain 
        $$
        \begin{array}{l}
          \bfS=[
            y^2 + (3*x^3 + x^2 + x + 1)*y + 2*x^3 + x^2 + x,\\
            y + 3*x^3 - x^2 + 1,\\
            3*x^3 - x^2,\\
            0
          ]
        \end{array}
        $$
        \item The last non nilpotent is 
          actually $B_1=\bfS[2]$ and the first one is $\bfS[3] = 3*x^3 - x^2$.
        \item While loop Line~\ref{enum:stillNil} ($j=3$) $S_{r_j} = \bfS[3]$ \\
          $\nilpotentFactor(S_{r_{j}} , x^4)== x^2, -1$  
          $\qquad (P_{new}=x^2, \ \  x^2 \mid \bfS[3],\ i_{new}=-1$)
        \item While loop Line~\ref{enum:stillNil} ($j=2$) $S_{r_j} = \bfS[2]=B_1$ \\
          $\nilpotentFactor(S_{r_{j}} , x^4)== 1, 1$ $(\qquad P_{new}=1$,
          $\coeff(B_1, i_{new}) = \coeff( B_1 ,  y ) = 1)$
        \item exit while loop and return $\subres(g_1,B_1,x^4) == \bfS, 2, 1, x^2$
      \end{itemize}
    \item Line~\ref{enum:WeierA}: $\WeierstrassMonic(B_1, x^4, 1) == B_1$
  \item Lines~\ref{enum:divbyp} $S = \bfS[3]/x^2 =  3*x-1 $
  \item Lines~\ref{enum:33}-\ref{enum:34}: $p^\ell = x^4/x^2 = x^2$, 
    $  \qquad \nilpotentFactor(S, x^2)== 1, 0$ ($i=0$) 
  \item Line~\ref{enum:WeierB}: $B_2=3 * x 1 $
  \item return $ B_1 ( = g_2)  , x^2 (= T_2') , 3 * x - 1 ( = B_2)  , x^2 ( = S_2)$
  \end{itemize}
  \item (Hensel Lifting, preparation) Line~\ref{enum:Gi}:
    $G_1 = g_1/ g_2 = b/B = y+x  \bmod \l x^2\r$
  \item Lines\ref{enum:xSres}-\ref{enum:precision}
      $\alpha = -(1+x), \ \beta =(1+x) $, $T_2 = x^3 \cdot x^2 =x^5, \quad $
    precision$=5$.
    \item (Hensel Lifting): Step 1: from $x^2$ to $x^4$:
      $$
       G_1^\star = y + 2*x^3 + x^2 + x,\quad g_2^\star  = y + x^3 + 1
      $$
      Step2: from $x^4$ to $x^8$:
      $$
      G_1^\star = y - 2*x^7 + x^6 + x^5 + x^4 + 2*x^3 + x^2 + x,
      \quad g_2^\star = y + 2*x^7 + x^5 + x^3 + 1
      $$
    \item Line~\ref{enum:Anext}: $g_2 = g_2^\star \bmod \l x^5 \r  = y + x^3 + 1$
    \item Lines~\ref{enum:block}-\ref{enum:52}: $\calD =[ G_1^\star \bmod \l x^3\r] = [y  + x^2 + x]$, $\Tree = [ x^3 , x^5]$, $\calC = [ b  ,  y  + x^3 + 1]$
    \item Line~\ref{enum:while}  ($i=2$): $\LargestFactor(g_2, 3*x-1 , x^5) ==$
      \begin{itemize}
      \item Line~\ref{enum:endb}     $\deg_y(3*x-1)=0,\qquad $ return $3*x-1, x^6,$ ``end''.
      \end{itemize}
      \item return $\calC =[ b, y +  x^3 + 1]$, $\calD=[y + x^2 + x]$, $\Tree=[x^3, x^5]$.     
\end{itemize}

According to Definition~\ref{def:gcdChain}, we have here $s=2$ blocks. And moreover $G_1 = \calD[1]$ and $G_2 = g_2 = \calC[2]$,
yielding the isomorphism:
$$
(\Q[x]/\l x^7\r)[y]/\l a,b\r \simeq (\Q[x]/\l x^3\r)[y]/\l y + x^2 + x \r \times (\Q[x]/\l x^5\r)[y]/\l y + x^3 + 1 \r   
$$
In this example, this decomposition coincides with the primary decomposition.

\subsubsection*{Generalization}
All steps of Algorithm~\ref{algo:largestFactor} ``largestFactor''
extends to  more than one variable, except the management of the first nilpotent subresultant.
To illustrate this difficulty,
let us consider a primary triangular set $T = (x^2 , y^2)$ of radical $(x,y)$,
and  some input polynomials $a$ and $b$ in $(k[x,y]/\l T\r)[z]$.
It may happen that a  subresultant is equal to say
$ x z + y $ (for example $a=z^2 + 2 xz - y$ and $b=z^2 + x z - 2 y$).
It is nilpotent and the iterated resultant criterion allows to detect it.
But there is no  way to ``remove'' the nilpotent part  as in the case of a polynomial of one variable.
To apply Weierstrass preparation theorem, the polynomial must not be nilpotent.
A solution consists of ``adding'' this polynomial to the coefficient ring. How this integration
shall be done requires more work and  will be investigated in the future.

\subsubsection*{Complexity}
The running-time of the algorithm is dominated by that of the  subresultant calls.
Indeed, all other subroutines are indeed based on classical algorithms which have a lower cost:
Hensel lifting, inversion of a truncated power series are endowed of fast algorithms;
And Weierstrass factorization is reduced to an Euclidean division.
The standard subresultant p.r.s. have a quadratic (operations in $k[x]/\l p^e\r$)
cost in the degrees of the input polynomials.
Therefore a running-time of $O( (e \deg(p) )^2 \cdot \deg(a)^2  )$
can be estimated using naive but realistic  algorithms..

We speculate that a fast version of 
our algorithm has a quasi-linear cost $O\tilde ((e \deg(p) \cdot \deg(a))$, 
by using  fast ``divide and conquer'' (a.k.a half-gcd).
This is quite challenging, and as one knows, is not reflected in any practical algorithm.
For now, we care on {\em feasability} with a view toward extensions to several variables.

%% \subsubsection*{Cofactors}
%% It would be interesting to consider cofactors in Isomorphism~\eqref{eq:isochain}. We mean polynomials
%% $u, v$ and $\{ \upsilon_j\}_{1\le j \le s}$, such that:
%% $$
%% a u + b v = \upsilon_1 ( H_1 + p^{e_1} h_1) + \cdots +  \upsilon_s (H_s + p^{e} h_s),
%% \quad \text{where:}
%% $$
%% $H_i + p^{e_i} h_i \equiv 1 \bmod \l G_i \r$ for $i=1,\ldots,s$;
%% $\deg_y(\upsilon_i) < \deg_y(G_i)$;
%% $\deg_y(H_i) \le \sum_{j\not= i} \deg_y (G_j)$.
%% The degree constraints on $u$ and $v$ %%are not as easy as usual to state here.
%% %% They 
%% should be made to insure  uniqueness while minimizing the degree.

%% The resultant is likely  possible  to compute at the cost of minor
%% modifications. At first sight,
%% one should keep track of the invertible factors obtained when
%% undergoing  Weierstrass factorization. Interestingly, if the complexity can
%% be made quasi-linear as mentionned above, this would match the complexity of~\cite{moroz2016fast}
%% but by a more straightforward algorithm.

%% \begin{ex}\label{ex:2}

%% \end{ex}

\bibliographystyle{plain}
%% \bibliography{c:/Users/Xav/GoogleDrive/main}
%% \bibliography{d:/GoogleDrive/main}
%%\bibliography{/home/xavier/Volume/GoogleDrive/main}

\appendix
\section*{Appendix}

This is a modification of Lang's proof~\cite[Theorem~9.1]{lang2002algebra} up to a minor point indicated below,
of the Weierstrass division, which is used in Corollary~\ref{cor:div}.
\begin{Lem}\label{lem:WeierDiv}
With the same notations and assumptions of
Proposition~\ref{prop:WeierFac}, given $g\in \fa[[X]]$, we can solve the equation $g = q f +r$
uniquely with $q\in \fa[[X]]$ and $r\in \fa[X]$ and $\deg r < k$.
\end{Lem}

\begin{proof}
Define the linear maps $\alpha: \fa[[X]] \to \fa[X]$, $\sum_i a_i X^i \mapsto \sum_{i=0}^{k-1} a_i X^i$
and $\tau: \fa[[X]] \to \fa[[X]] $, $\sum_i a_i X^i \mapsto \sum_{i} a_{k+i}X^i$.
It is clear that  $b\in \fa[X]$ has degree $< k$ if and only if  $\tau(b)=0$ or $\alpha(b)=b$.
Therefore if $ g = q f +r$, $\deg(r)<k$ then $\tau(g)= \tau(q f)$.
Moreover $\tau (X^k h) = h$ and $ h = \tau(h) X^k + \alpha(h)$ for any $h\in \fa[[X]]$.
Thus
$$
\tau(g) = \tau(q (\alpha(f) + X^k \tau(f))) = \tau(q \alpha(f) )  +\tau(f) q.
$$
Let $Z= \tau(f)q$. Notice that $\tau(f) = f_k +X m'$ with $f_k\in \fa^\star$ and $m'\in \fm \fa[[X]]$ by assumption, hence $\tau(f)$ is
invertible in $\fa[[X]]$. The equation above can be rewritten: $\tau(g) = \tau( Z \frac{\alpha(f)}{\tau(f)}) + Z$.
Being able to solve  this equation in $Z$ uniquely gives $\tau(f) q$, hence $q$, hence $r=g -  f q$ in a unique way.

To do it, the proof differs slightly from that
of ~\cite[Theorem~9.2]{lang2002algebra} to the following point. The reason why the image of the map below 
$$
\tau \circ \frac{\alpha(f)}{\tau(f)}:\fa[[X]] \to \fm \fa[[X]]
$$
is $\fm \fa [[X]$ is because $\tau(f)^{-1}= f_k^{-1} \sum_{\ell\ge 0} (-1)^\ell (  m X)^\ell$ for an $m\in \fm \fa[[X]]$
and hence
$\alpha(f) \tau(f)^{-1} = \alpha(f) f_k^{-1} + \alpha(f) M$ where $M\in \fm \fa[[X]]$.
Therefore $\tau( \alpha(f) \tau(f)^{-1}) \in \fm \fa[[X]]$ since $\deg(\alpha(f) f_{k-1}) = \deg (\alpha(f)) \le k-1$. Thus  for
any power series $h\in \fa[[X]]$,
$$
\left( \tau \circ \frac{\alpha(f)}{\tau(f)}\right)  (h) \in \fm \fa [[X]].
$$
Now $\fa[[X]]$ being a complete ring,
$I + \tau \circ \frac{\alpha(f)}{\tau(f)} $ is invertible and  $Z = (I + \tau \circ \frac{\alpha(f)}{\tau(f)})^{-1} (\tau(g)) $ is determined in a unique way.
\hfill $\Box$ \end{proof}

%% Second a technical but easy result that:
%% \begin{lemma}\label{lem:idealproduct}
%% Consider an equality of ideals  $J = I_1,\cdots I_s$ in a ring $A$.
%% Given another ideal $K$, one has: $ J \cap K \supset I_1 \cdots I_{s-1} (I\cap K)$.
%% \end{lemma}\
%% \begin{proof}
%% Elements of the r.h.s. are finite sums of terms of type $a_1 \cdots a_s$ with 
%% $a_i \in I_i$ and $a_s \in K \cap I_s$. Hence all terms belong to $K$ ( as well as $J$
%% by assumption).
%% \hfill $\Box$ \end{proof}
\end{document}